\documentclass[letterpaper,11pt]{article}

\usepackage{times} 
\usepackage{fullpage} 
\usepackage{amsmath} 
\usepackage{amssymb} 
\usepackage{amsthm} 
\usepackage{bbm} 

\newcommand{\inner}[2]{\langle #1 , #2\rangle}
\newcommand{\Inner}[2]{\left\langle #1 , #2\right\rangle}
\newcommand{\defeq}{\stackrel{\smash{\text{\tiny def}}}{=}}
\newcommand{\trans}{{\scriptstyle\mathsf{T}}}

\DeclareMathOperator{\convex}{convex-hull}
\DeclareMathOperator{\spn}{span}

\newcommand{\kprod}[3]{#1_{#2\dots#3}}

\newtheorem{theorem}{Theorem}
\newtheorem{lemma}[theorem]{Lemma}
\newtheorem{prop}[theorem]{Proposition}
\newtheorem{corollary}{Corollary}[theorem]
\theoremstyle{definition}
\newtheorem{remark}[corollary]{Remark}
\newtheorem{notation}{Notation for Theorem}

\newcommand{\br}[1]{[#1]}
\newcommand{\cBr}[1]{\left[#1\right]}
\newcommand{\pa}[1]{(#1)}
\newcommand{\Pa}[1]{\left(#1\right)}
\newcommand{\set}[1]{\{#1\}}
\newcommand{\Set}[1]{\left\{#1\right\}}

\newcommand{\Jamiolkowski}{J}
\newcommand{\jam}[1]{\Jamiolkowski\pa{#1}}

\DeclareMathOperator{\trace}{Tr}
\newcommand{\ptr}[2]{\trace_{#1}\pa{#2}}
\newcommand{\Ptr}[2]{\trace_{#1}\Pa{#2}}
\newcommand{\tr}[1]{\ptr{}{#1}}

\newcommand{\tinyspace}{\mspace{1mu}}

\newcommand{\cAbs}[1]{\left|\tinyspace#1\tinyspace\right|}
\newcommand{\norm}[1]{\lVert\tinyspace#1\tinyspace\rVert}
\newcommand{\Norm}[1]{\left\lVert\tinyspace#1\tinyspace\right\rVert}
\newcommand{\fnorm}[1]{\norm{#1}_{\mathrm{F}}}
\newcommand{\Fnorm}[1]{\Norm{#1}_{\mathrm{F}}}
\newcommand{\tnorm}[1]{\norm{#1}_{\trace}}

\newcommand{\fontmapset}{\mathbf} 

\newcommand{\mset}[2]{\fontmapset{#1}\pa{#2}}

\newcommand{\lin}[1]{\mset{L}{#1}}

\newcommand{\her}[1]{\mset{H}{#1}}

\newcommand{\pos}[1]{\mset{H^+}{#1}}

\newcommand{\identity}{\mathbbm{1}}
\newcommand{\idsup}[1]{\identity_{#1}}

\def\noisy{\Delta}

\def\ot{\otimes}

\def\cA{\mathcal{A}}

\def\cE{\mathcal{E}}

\def\cV{\mathcal{V}}

\def\cX{\mathcal{X}}
\def\cY{\mathcal{Y}}
\def\cZ{\mathcal{Z}}

\def\bK{\mathbf{K}}

\def\bQ{\mathbf{Q}}

\def\bS{\mathbf{S}}

\begin{document}

\title{Properties of Local Quantum Operations with Shared Entanglement}

\author{
  Gus Gutoski\thanks{
    Institute for Quantum Computing
    and School of Computer Science,
    University of Waterloo,
    Waterloo, Ontario, Canada.
  }
}

\date{April 7, 2009}

\maketitle

\begin{abstract}

Multi-party local quantum operations with shared quantum entanglement or shared
classical randomness are studied.
The following facts are established:
\begin{itemize}

\item
There is a ball
of local operations with shared randomness lying within the space
spanned by the no-signaling operations and centred at the completely noisy
channel.

\item
The existence of the ball of local operations with shared randomness is
employed to prove that the weak membership problem for local operations
with shared entanglement is strongly NP-hard.

\item
Local operations with shared entanglement are
characterized in terms of linear functionals that are
``completely'' positive on a certain cone $K$ of separable Hermitian operators,
under a natural notion of complete positivity appropriate to that cone.
Local operations with shared randomness
(but not entanglement) are also characterized in terms of linear functionals
that are merely positive on that same cone $K$.

\item
Existing characterizations of
no-signaling operations are generalized
to the multi-party setting
and recast in terms of the Choi-Jamio\l kowski
representation for quantum super-operators.
%
It is noted that the standard nonlocal box is an example of a no-signaling
operation that is separable, yet cannot be implemented by local operations with shared
entanglement.
%

\end{itemize}

\end{abstract}

\section{Introduction}


A quantum operation that maps a finite-dimensional bipartite quantum input system
$\cX_1\ot\cX_2$ to another finite-dimensional bipartite quantum output system
$\cA_1\ot\cA_2$ is a
\emph{local operation with shared entanglement (LOSE)}
if it can be approximated to arbitrary precision by quantum operations
$\Lambda$ for which there exists a fixed quantum state $\sigma$
and quantum operations $\Psi_1,\Psi_2$
such that the action of $\Lambda$ on each input state $\rho$ is given by
\[ \Lambda\pa{\rho} = \Pa{ \Psi_1\ot\Psi_2 }\pa{\rho\ot\sigma}. \]
Specifically, $\sigma$ is a state of some finite-dimensional bipartite
quantum auxiliary system $\cE_1\ot\cE_2$ and the quantum operations $\Psi_i$ map
the system $\cX_i\ot\cE_i$ to the system $\cA_i$ for each $i=1,2$.
This arrangement is illustrated in Figure \ref{fig:LO}.

\begin{figure}[h]
  \hrulefill \vspace{2mm}
  \begin{center}
\setlength{\unitlength}{2072sp}%
\begin{picture}(4140,4617)(846,-4123)
\thinlines
\put(2881,-1051){\line( 1, 0){360}}
\put(2881,-1142){\line( 1, 0){360}}
\put(2881,-1231){\line( 1, 0){360}}
\put(2881,-1322){\line( 1, 0){360}}
\put(2881,-1411){\line( 1, 0){360}}
\put(2881,-961){\line( 1, 0){360}}
\put(2881,-871){\line( 1, 0){360}}
\put(2881,-1501){\line( 1, 0){360}}
\put(2881,-1591){\line( 1, 0){360}}
\put(2881,-1681){\line( 1, 0){360}}
\put(2881,-2401){\line( 1, 0){360}}
\put(2881,-2492){\line( 1, 0){360}}
\put(2881,-2581){\line( 1, 0){360}}
\put(2881,-2672){\line( 1, 0){360}}
\put(2881,-2761){\line( 1, 0){360}}
\put(2881,-2311){\line( 1, 0){360}}
\put(2881,-2221){\line( 1, 0){360}}
\put(2881,-2851){\line( 1, 0){360}}
\put(2881,-2941){\line( 1, 0){360}}
\put(2881,-3031){\line( 1, 0){360}}
\put(1341,-3391){\line( 1, 0){1900}}
\put(1341,-3481){\line( 1, 0){1900}}
\put(1341,-3571){\line( 1, 0){1900}}
\put(1341,-3661){\line( 1, 0){1900}}
\put(1341,-3751){\line( 1, 0){1900}}
\put(1341,-3301){\line( 1, 0){1900}}
\put(1341,-3211){\line( 1, 0){1900}}
\put(4141,-2941){\line( 1, 0){540}}
\put(4141,-3031){\line( 1, 0){540}}
\put(4141,-3121){\line( 1, 0){540}}
\put(4141,-3211){\line( 1, 0){540}}
\put(4141,-3301){\line( 1, 0){540}}
\put(4141,-2851){\line( 1, 0){540}}
\put(4141,-2761){\line( 1, 0){540}}
\put(4141,-2671){\line( 1, 0){540}}
\put(4141,-2581){\line( 1, 0){540}}
\put(4141,-3391){\line( 1, 0){540}}
\put(4141,-871){\line( 1, 0){540}}
\put(4141,-961){\line( 1, 0){540}}
\put(4141,-1051){\line( 1, 0){540}}
\put(4141,-1141){\line( 1, 0){540}}
\put(4141,-1231){\line( 1, 0){540}}
\put(4141,-781){\line( 1, 0){540}}
\put(4141,-691){\line( 1, 0){540}}
\put(4141,-601){\line( 1, 0){540}}
\put(4141,-511){\line( 1, 0){540}}
\put(4141,-1321){\line( 1, 0){540}}
\put(1341,-331){\line( 1, 0){1900}}
\put(1341,-421){\line( 1, 0){1900}}
\put(1341,-511){\line( 1, 0){1900}}
\put(1341,-601){\line( 1, 0){1900}}
\put(1341,-691){\line( 1, 0){1900}}
\put(1341,-241){\line( 1, 0){1900}}
\put(1341,-151){\line( 1, 0){1900}}
\put(3241,-3931){\framebox(900,1890){$\Psi_2$}}
\put(1701,-4111){\dashbox{57}(2620,4320){}}
\put(3241,-1861){\framebox(900,1890){$\Psi_1$}}
\put(2501,-1446){\makebox(0,0)[lb]{$\cE_1$}}
\put(2501,-2746){\makebox(0,0)[lb]{$\cE_2$}}
\put(4771,-1086){\makebox(0,0)[lb]{}}
\put(4771,-3156){\makebox(0,0)[lb]{}}
\put(911,-586){\makebox(0,0)[lb]{$\cX_1$}}
\put(911,-3646){\makebox(0,0)[lb]{$\cX_2$}}
\put(4800,-1086){\makebox(0,0)[lb]{$\cA_1$}}
\put(4800,-3146){\makebox(0,0)[lb]{$\cA_2$}}
\put(2871,389){\makebox(0,0)[lb]{$\Lambda$}}
\put(1911,-2060){$\sigma\left\{\rule[-11.25mm]{0mm}{0mm}\right.$}
\end{picture}%
  \end{center}
    \caption{A two-party local quantum operation $\Lambda$
    with shared entanglement or randomness.
    The local operations are represented by $\Psi_1,\Psi_2$;
    the shared entanglement or randomness by $\sigma$.}
    \label{fig:LO}
  \hrulefill
\end{figure}
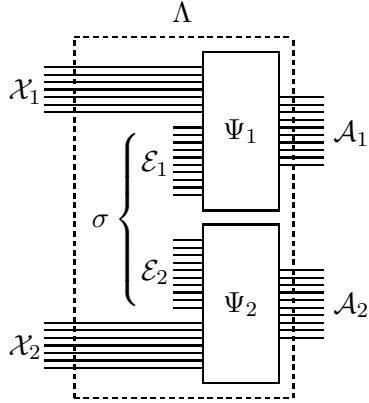


Intuitively, the two operations $\Psi_1,\Psi_2$ represent the actions
of two distinct parties who cannot communicate once they receive their portions
of the input state $\rho$.
The parties can, however, arrange ahead of time to share distinct portions of
some distinguished auxiliary state $\sigma$.
This state---and any entanglement contained therein---may be used by the
parties to correlate their separate operations on $\rho$.
This notion is easily generalized to a multi-party setting and the results presented herein extend to this generalized setting with minimal complication.

LOSE operations
are of particular interest in the quantum information community in part because
any physical operation jointly implemented by spatially separated parties who
obey both quantum mechanics and relativistic causality must necessarily be of
this form.
Moreover, it is notoriously difficult to say anything meaningful about
this class of operations, despite its simple definition.

One source of such difficulty stems from the fact that
there exist two-party LOSE operations with the fascinating property that they
cannot be implemented with any finite amount of shared
entanglement \cite{LeungT+08}.
(Hence the allowance for arbitrarily close approximations in the preceding
definition of LOSE operations.)
This difficulty has manifested itself quite prominently in the study of
two-player co-operative games:
classical games characterize NP,\footnote{
As noted in Ref.~\cite{KempeK+07}, this characterization follows from the
PCP Theorem \cite{AroraL+98,AroraS98}.
}
%
%
%
%
whereas quantum games with shared entanglement are not even known to be
computable. 
Within the context of these games, interest has focussed largely on special
cases of LOSE operations
\cite{KobayashiM03,CleveH+04,Wehner06,CleveS+07,CleveGJ07,KempeR+07},
but progress has been made recently in the general case
\cite{KempeK+07,KempeK+08,LeungT+08,DohertyL+08,NavascuesP+08}.
In the physics literature, LOSE operations are often discussed in the context of
\emph{no-signaling} operations
\cite{BeckmanG+01,EggelingSW02,PianiH+06}.

In the present paper some light is shed on the general class of multi-party LOSE operations,
as well as sub-class of operations derived therefrom as follows:
if the quantum state $\sigma$ shared among the different parties of a
LOSE operation $\Lambda$ is \emph{separable} then it is easy to see that
the parties could implement $\Lambda$ using only shared \emph{classical randomness}.
Quantum operations in this sub-class are therefore called
\emph{local operations with shared randomness (LOSR)}.

Several distinct results are established in the present paper,
many of which mirror some existing result pertaining to separable quantum
states.
What follows is a brief description of each result together with its analogue
from the literature on separable states where appropriate.

\subsubsection*{Ball around the identity}

\begin{description}

\item \emph{Prior work on separable quantum states.}
If $A$ is a Hermitian operator
acting on a $d$-dimensional bipartite space
and whose Frobenius norm at most 1
then the perturbation $I\pm A$ of the identity
represents an (unnormalized) bipartite separable quantum state.
In other words, there is a ball of (normalized) bipartite separable
states with radius $\frac{1}{d}$ in Frobenius norm centred at
the completely mixed state $\frac{1}{d}I$
\cite{Gurvits02,GurvitsB02}.

A similar ball exists in the multipartite case, but with smaller radius.
In particular, there is a ball of $m$-partite separable $d$-dimensional states
with radius $\Omega\Pa{2^{-m/2}d^{-1}}$ in Frobenius norm
centred at the completely mixed state
\cite{GurvitsB03}.
Subsequent results offer some improvements on this radius
\cite{Szarek05,GurvitsB05,Hildebrand05}.

\item \emph{Present work on local quantum operations.}
Analogous results are proven in Sections \ref{sec:gen} and \ref{sec:balls}
for multi-party LOSE and LOSR operations.
Specifically, if $A$ is a Hermitian operator
acting on a $n$-dimensional $m$-partite space
and whose Frobenius norm scales as
$O\Pa{2^{-m}n^{-3/2}}$ then $I\pm A$ is the
Choi-Jamio\l kowski representation of an (unnormalized) $m$-party LOSR operation.
As the unnormalized completely noisy channel
\[\noisy:X\mapsto\tr{X}I\]
is the unique quantum operation whose Choi-Jamio\l kowski representation equals the
identity, it follows that there is a ball of $m$-party LOSR operations
(and hence also of LOSE operations)
with radius $\Omega\Pa{2^{-m}n^{-3/2}d^{-1}}$ in Frobenius norm
centred at the completely noisy channel $\frac{1}{d}\noisy$.
(Here the normalization factor $d$ is the dimension of the output system.)

The perturbation $A$ must lie in the space spanned by
Choi-Jamio\l kowski representations of the no-signaling operations.
Conceptual implications of this technicality are discussed in
Remark \ref{rem:balls:noise} of
Section \ref{sec:balls}.
No-signaling operations are discussed in the present paper,
as summarized below.

\item \emph{Comparison of proof techniques.}
Existence of this ball of LOSR operations is established via elementary linear
algebra.
By contrast, existence of the ball of separable states was originally
established via a delicate combination of
(i)
the fundamental characterizations of separable states in terms of positive
super-operators (described in more detail below), together with
(ii)
nontrivial norm inequalities for these super-operators.

Moreover, the techniques presented herein
for LOSR operations
are of sufficient generality to immediately imply a ball of separable states
without the need for the aforementioned characterizations or
their accompanying norm inequalities.
This simplification comes in spite of the inherently more complicated nature of
LOSR operations as compared to separable states.
It should be noted, however, that the ball of separable states implied by the
present work is smaller than the ball established in prior work by a factor of
$2^{-m/2}d^{-3/2}$.


\end{description}

\subsubsection*{Weak membership problems are NP-hard}

\begin{description}

\item \emph{Prior work on separable quantum states.}
The weak membership problem asks,
\begin{quote}
  ``Given a description of a quantum state $\rho$
  and an accuracy parameter $\varepsilon$,
  is $\rho$ within distance $\varepsilon$ of a separable state?''
\end{quote}
This problem was proven strongly NP-complete under oracle (Cook) reductions by
Gharibian \cite{Gharibian08},
who built upon the work of Gurvits \cite{Gurvits02} and Liu \cite{Liu07}.
In this context, ``strongly NP-complete'' means that the problem remains
NP-complete even when the accuracy parameter $\varepsilon=1/s$ is given in
unary as $1^s$.

NP-completeness of the weak membership problem was originally established by Gurvits \cite{Gurvits02}.
The proof consists of a NP-completeness result for the weak \emph{validity} problem---a decision version of linear optimization over separable states---followed by an application of the 
Yudin-Nemirovski\u\i{} Theorem~\cite{YudinN76, GrotschelL+88}, which provides an oracle-polynomial-time reduction from weak validity to weak membership for general convex sets.

As a precondition of the Theorem, the convex set in question---in this case, the
set of separable quantum states---must contain a reasonably-sized ball.
In particular, this NP-completeness result relies crucially upon the existence
of the aforementioned ball of separable quantum states.

For \emph{strong} NP-completeness, it is necessary to employ a specialized 
``approximate'' version of the Yudin-Nemirovski\u\i{} Theorem
due to Liu~\cite[Theorem 2.3]{Liu07}.

\item \emph{Present work on local quantum operations.}
In Section \ref{sec:hard} it is proved that the weak membership problems for
LOSE and LOSR operations are both strongly NP-hard under oracle reductions.
The result for LOSR operations follows trivially from Gharibian
(just take the input spaces to be empty),
but it is unclear how to obtain the result for LOSE operations without invoking
the contributions of the present paper.

The proof begins by observing that the weak validity problem for LOSE
operations is merely a two-player quantum game in disguise and hence is
strongly NP-hard \cite{KempeK+07}.
The hardness result for the weak \emph{membership} problem is then obtained via
a Gurvits-Gharibian-style application of Liu's version of the
Yudin-Nemirovski\u\i{} Theorem, which of course depends upon the existence of
the ball revealed in Section \ref{sec:balls}.

\end{description}

\subsubsection*{Characterization in terms of positive super-operators}

\begin{description}

\item \emph{Prior work on separable quantum states.}
A quantum state $\rho$ of a bipartite system $\cX_1\ot\cX_2$
is separable if and only if the operator
\[ \Pa{\Phi\otimes\identity_{\cX_2}}\pa{\rho} \]
is positive semidefinite whenever the super-operator
$\Phi$ is positive.
This fundamental fact was first proven in 1996 by
Horodecki \emph{et al.}~\cite{HorodeckiH+96}.

The multipartite case reduces inductively to the bipartite case:
the state $\rho$ of an $m$-partite system
is separable if and only if
\( \Pa{\Phi\otimes\identity}\pa{\rho} \)
is positive semidefinite whenever the super-operator
$\Phi$ is positive on
$(m-1)$-partite separable operators \cite{HorodeckiH+01}.

\item \emph{Present work on local quantum operations.}
In Section \ref{sec:char} it is proved that a multi-party quantum operation
$\Lambda$ is a LOSE operation if and only if
\[ \varphi\pa{\jam{\Lambda}}\geq 0 \]
whenever the linear functional $\varphi$ is
``completely'' positive on a certain cone of separable Hermitian operators,
under a natural notion of complete positivity appropriate to that cone.
A characterization of LOSR operations is obtained by replacing
\emph{complete} positivity of $\varphi$ with mere positivity
on that same cone.
Here $\jam{\Lambda}$ denotes the Choi-Jamio\l kowski representation of the
super-operator $\Lambda$.

The characterizations presented in Section \ref{sec:char} do not rely
upon discussion in previous sections.
This independence contrasts favourably with prior work on separable quantum
states, wherein the existence of the ball around the completely mixed state
(and the subsequent NP-hardness result)
relied crucially upon the characterization of separable states in terms of
positive super-operators.

\end{description}

\subsubsection*{No-signaling operations}

A quantum operation is \emph{no-signaling} if it cannot be used by spatially
separated parties to violate relativistic causality.
By definition, every LOSE operation is a no-signaling operation.
Moreover, there exist no-signaling operations that are not LOSE operations.
Indeed, it is noted in Section \ref{sec:no-sig} that the standard nonlocal
box of Popescu and Rohrlich \cite{PopescuR94}
is an example of a no-signaling operation that is separable
(in the sense of Rains \cite{Rains97}),
yet it is not a LOSE operation.

Two characterizations of no-signaling operations
are also discussed in Section \ref{sec:no-sig}.
These characterizations were first established somewhat implicitly for the
bipartite case in Beckman \emph{et al.}~\cite{BeckmanG+01}.
The present paper generalizes these characterizations to the multi-party
setting and recasts them in terms of the Choi-Jamio\l kowski
representation for quantum super-operators.

\subsubsection*{Open problems}

Several interesting open problems are
brought to light by the present work.
These problems are discussed in Section \ref{sec:conclusion}.

\section{General Results on Separable Operators} 
\label{sec:gen}

It is not difficult to see that a quantum operation is a LOSR
operation if and only if it can be written as a convex combination of product
quantum operations of the form $\Phi_1\otimes\Phi_2$.
(This observation is proven in Proposition \ref{prop:LOSR-convex} of Section \ref{sec:balls:defs}.)
In order to be quantum operations, the super-operators $\Phi_1,\Phi_2$
must be completely positive and trace-preserving.
In terms of Choi-Jamio\l kowski representations, complete positivity holds if
and only if the operators $\jam{\Phi_1}$, $\jam{\Phi_2}$ are positive
semidefinite.
The trace-preserving condition is characterized by simple linear
constraints on $\jam{\Phi_1}$, $\jam{\Phi_2}$, the details of which are
deferred until Section \ref{sec:balls}.
It suffices for now to observe that these constraints induce
subspaces $\bQ_1$, $\bQ_2$ of Hermitian operators within which
$\jam{\Phi_1}$, $\jam{\Phi_2}$ must lie.
In other words, the quantum operation $\Lambda$
is a LOSR operation if and only if
$\jam{\Lambda}$ lies in the set
\[
  \convex \Set{
    X\otimes Y : X\in\bQ_1, Y\in\bQ_2,
    \textrm{ and $X,Y$ are positive semidefinite}
  }.
\]
As the unnormalized completely noisy channel $\noisy$ has $\jam{\noisy}=I$,
a ball of LOSR operations around this
channel may be established by showing that every operator in the product space
$\bQ_1\otimes\bQ_2$ and
close enough to the identity lies in the above set.
Such a fact is established in the present section.
Indeed, this fact is shown to hold not only for the specific choice
$\bQ_1,\bQ_2$ of subspaces, but for \emph{every} choice
$\bS_1,\bS_2$ of subspaces that contain the identity.

Choosing $\bS_1$ and $\bS_2$ to be full spaces of Hermitian operators yields
an alternate (and simpler) proof of the existence of a ball of separable quantum
states surrounding the completely mixed state and, consequently, of the
NP-completeness of separability testing.
However, the ball of separable states implied by the present work is not as
large as that exhibited by Gurvits and Barnum \cite{GurvitsB02, GurvitsB03}.

Due to the general nature of this result, discussion in this section is
abstract---applications to quantum information are deferred until
Section \ref{sec:balls}.
The results presented herein were inspired by Chapter 2 of
Bhatia~\cite{Bhatia07}.

\subsection{Linear Algebra: Review and Notation}

For an arbitrary operator $X$ the standard operator norm of $X$ is denoted
$\norm{X}$
and the standard conjugate-transpose (or \emph{Hermitian conjugate}) is denoted
$X^*$.
For operators $X$ and $Y$, the standard Hilbert-Schmidt inner product is given
by \[\inner{X}{Y}\defeq\ptr{}{X^*Y}.\]
The standard Euclidean 2-norm for vectors is also a norm on operators.
It is known in this context as the \emph{Frobenius norm} and is given by
\[ \fnorm{X}\defeq\sqrt{\inner{X}{X}}. \]
Discussion in the present paper often involves finite sequences of operators,
spaces, and so on.
As such, it is convenient to adopt a shorthand notation for the Kronecker
product of such a sequence.
For example, if $\bS_1,\dots,\bS_m$ are spaces of operators then
\[
  \kprod{\bS}{i}{j} \defeq
  \bS_i \otimes \cdots \otimes \bS_j
\]
denotes their Kronecker product
for integers $1\leq i\leq j\leq m$.
A similar notation applies to operators.

For any space $\bS$ of operators, $\bS^+$ denotes the intersection of
$\bS$ with the positive semidefinite cone.
An element $X$ of the product space
$\kprod{\bS}{1}{m}$
is said to be
\emph{$(\bS_1;\dots;\bS_m)$-separable} if $X$ can be written as a
convex combination of
product operators of the form
$P_1\otimes\cdots\otimes P_m$ where each $P_i\in\bS_i^+$
is positive semidefinite.
The set of $(\bS_1;\dots;\bS_m)$-separable operators forms a cone inside
$\Pa{\kprod{\bS}{1}{m}}^+$.
By convention, the use of bold font is reserved for sets of operators.

\subsection{Hermitian Subspaces Generated by Separable Cones}
\label{sec:gen:sep}

Let $\bS$ be any subspace of Hermitian operators that contains the identity.
The cone $\bS^+$ always \emph{generates} $\bS$, meaning that each element of
$\bS$ may be written as a difference of two elements of $\bS^+$.
As proof, choose any $X\in\bS$ and let
\[ X^\pm = \frac{ \norm{X} I \pm X }{2}. \]
It is clear that $X=X^+ - X^-$ and that $X^\pm\in\bS^+$
(using the fact that $I\in\bS$.)
Moreover, it holds that $\norm{X^\pm}\leq\norm{X}$ for this particular choice of
$X^\pm$.

In light of this observation, one might wonder whether it could be extended in
product spaces to separable operators.
In particular, do the
$(\bS_1;\dots;\bS_m)$-separable operators generate the product space
$\kprod{\bS}{1}{m}$?
If so, can elements of $\kprod{\bS}{1}{m}$ be generated by
$(\bS_1;\dots;\bS_m)$-separable operators with bounded norm?
The following theorem answers these two questions in the affirmative.

\begin{theorem} \label{thm:sep-bound}

  Let $\bS_1,\dots,\bS_m$ be subspaces of
  Hermitian operators---all of which contain the identity---and let
  $n=\dim\pa{\kprod{\bS}{1}{m}}$.
  Then every element $X\in\kprod{\bS}{1}{m}$ may be written
  $X = X^+ - X^-$ where $X^\pm$ are $(\bS_1;\dots;\bS_m)$-separable
  with $\norm{X^\pm}\leq 2^{m-1}\sqrt{n}\fnorm{X}$.

\end{theorem}

\begin{proof}

First, it is proven that there is an orthonormal basis $\mathbf{B}$ of
$\kprod{\bS}{1}{m}$ with the property that every element $E\in \mathbf{B}$
may be written $E = E^+ - E^-$ where $E^\pm$ are
$(\bS_1;\dots;\bS_m)$-separable
with $\norm{E^\pm}\leq 2^{m-1}$.
The proof is by straightforward induction on $m$.
The base case $m=1$ follows immediately from the earlier observation that
every element $X\in\bS_1$ is generated by some $X^\pm\in\bS_1^+$ with
$\norm{X^\pm}\leq \norm{X}$.
In particular, any element $E=E^+-E^-$ of any orthonormal basis of $\bS_1$ has
$\norm{E^\pm}\leq\norm{E}\leq\fnorm{E}=1=2^0$.

In the general case, the induction hypothesis states that there is
an orthonormal basis $\mathbf{B}'$ of $\kprod{\bS}{1}{m}$ with the desired
property.
Let $\mathbf{B}_{m+1}$ be any orthonormal basis of $\bS_{m+1}$.
As in the base case, each $F\in \mathbf{B}_{m+1}$ is generated by some
$F^\pm\in\bS_{m+1}^+$ with $\norm{F^\pm}\leq \norm{F} \leq 1$.
Define the orthonormal basis
$\mathbf{B}$ of $\kprod{\bS}{1}{m+1}$ to consist of all product operators
of the form $E\otimes F$ for $E\in \mathbf{B}'$ and $F\in \mathbf{B}_{m+1}$.
Define
\begin{align*}
  K^+ &\defeq \cBr{E^+\otimes F^+} + \cBr{E^-\otimes F^-},\\
  K^- &\defeq \cBr{E^+\otimes F^-} + \cBr{E^-\otimes F^+}.
\end{align*}
It is clear that $E\otimes F=K^+-K^-$ and $K^\pm$ are
$(\bS_1;\dots;\bS_{m+1})$-separable.
Moreover,
\[
  \Norm{K^+} \leq \Norm{E^+\otimes F^+} + \Norm{E^-\otimes F^-} \leq
  \cBr{2^{m-1}\times 1} + \cBr{2^{m-1}\times 1} = 2^m.
\]
A similar computation yields $\Norm{K^-}\leq 2^m$, which establishes the
induction.

Now, let $X\in\kprod{\bS}{1}{m}$ and let $x_j\in\mathbb{R}$ be the unique
coefficients of $X$ in the aforementioned orthonormal basis
$\mathbf{B}=\set{E_1,\dots,E_n}$.
Define
\begin{align*}
  X^+ &\defeq \sum_{j\::\:x_j>0} x_j E_j^+ - \sum_{j\::\:x_j<0} x_j E_j^-,\\
  X^- &\defeq \sum_{j\::\:x_j>0} x_j E_j^- - \sum_{j\::\:x_j<0} x_j E_j^+.
\end{align*}
It is clear that $X=X^+-X^-$ and $X^\pm$ are
$(\bS_1;\dots;\bS_m)$-separable.
Employing the triangle inequality and the vector norm inequality
$\norm{x}_1\leq\sqrt{n}\norm{x}_2$, it follows that
\[
  \Norm{X^+} \leq 2^{m-1} \sum_{j=1}^n \cAbs{x_j} \leq
  2^{m-1} \sqrt{n} \sqrt{ \sum_{j=1}^n \cAbs{x_j}^2 } =
  2^{m-1} \sqrt{n} \Fnorm{X}.
\]
A similar computation yields $\Norm{X^-}\leq 2^{m-1}\sqrt{n}\fnorm{X}$,
which completes the proof.
\end{proof}

\subsection{Separable Operators as a Subtraction from the Identity}

At the beginning of Section \ref{sec:gen:sep} it was observed that
$\norm{X}I-X$ lies in $\bS^+$ for all $X\in\bS$.
One might wonder whether more could be expected of $\norm{X}I-X$ than mere
positive semidefiniteness.
For example, under what conditions is $\norm{X}I-X$ a
$\pa{\bS_1;\dots;\bS_m}$-separable operator?
The following theorem provides three such conditions.
Moreover, the central claim of this section is established by this theorem.

\begin{theorem} \label{thm:identity-sep}

Let $\bS_1,\dots,\bS_m$ be subspaces of Hermitian operators---all of which
contain the identity---and let $n=\dim\pa{\kprod{\bS}{1}{m}}$.
The following hold:
\begin{enumerate}

\item \label{item:identity-sep:1}
  \( \Norm{P}I - P \)
  is $\pa{\bS_1;\dots;\bS_m}$-separable whenever $P$ is a product operator
  of the form $P=P_1\otimes\cdots\otimes P_m$ where each $P_i\in\bS_i^+$.

\item \label{item:identity-sep:2}
  \( \cBr{n+1}\Norm{Q}I - Q \) is
  $\pa{\bS_1;\dots;\bS_m}$-separable whenever $Q$ is
  $\pa{\bS_1;\dots;\bS_m}$-separable.

\item \label{item:identity-sep:3}
  \( 2^{m-1}\sqrt{n}\cBr{n+1}\Fnorm{X} I - X \) is
  $\pa{\bS_1;\dots;\bS_m}$-separable for every
  $X\in\kprod{\bS}{1}{m}$.

\end{enumerate}

\end{theorem}

\begin{proof}
The proof of item \ref{item:identity-sep:1} is an easy but notationally
cumbersome induction on $m$.
The base case $m=1$ was noted at the beginning of Section \ref{sec:gen:sep}.
For the general case,
it is convenient to let $\kprod{I}{1}{m}$, $I_{m+1}$, and $\kprod{I}{1}{m+1}$
denote the identity elements of
$\kprod{\bS}{1}{m}$, $\bS_{m+1}$, and $\kprod{\bS}{1}{m+1}$, respectively.
The induction hypothesis states that the operator
\( S \defeq \norm{\kprod{P}{1}{m}} \kprod{I}{1}{m} - \kprod{P}{1}{m} \)
is $\pa{\bS_1;\dots;\bS_m}$-separable.
Just as in the base case, we know
\( S' \defeq \norm{P_{m+1}}I_{m+1}-P_{m+1} \)
lies in $\bS_{m+1}^+$.
Isolating the identity elements, these expressions may be rewritten
\begin{align*}
  \kprod{I}{1}{m} &=
    \frac{1}{ \Norm{\kprod{P}{1}{m}} } \cBr{\kprod{P}{1}{m} + S} \\
  I_{m+1} &=
    \frac{1}{ \Norm{P_{m+1}} } \cBr{ P_{m+1} + S' }.
\end{align*}
Taking the Kronecker product of these two equalities and rearranging the terms
yields
\[
  \Norm{\kprod{P}{1}{m+1}} \kprod{I}{1}{m+1} - \kprod{P}{1}{m+1} =
  \cBr{ \kprod{P}{1}{m} \otimes S' } +
  \cBr{ S \otimes P_{m+1} } +
  \cBr{ S \otimes S' }.
\]
The right side of this expression is clearly a
$\pa{\bS_1;\dots;\bS_{m+1}}$-separable operator;
the proof by induction is complete.

Item \ref{item:identity-sep:2} is proved as follows.
By Carath\'eodory's Theorem, every $(\bS_1;\dots;\bS_m)$-separable operator
$Q$ may be written as a sum of no more than $n+1$ product operators.
In particular,
\[ Q = \sum_{j=1}^{n+1} P_{1,j}\otimes\cdots\otimes P_{m,j} \]
where each $P_{i,j}$ is an element of $\bS_i^+$.
As each term in this sum is positive semidefinite, it holds that
\( \norm{Q} \geq \norm{P_{1,j}\otimes\cdots\otimes P_{m,j}} \)
for each $j$.
Item \ref{item:identity-sep:1} implies that the sum
\[
  \sum_{j=1}^{n+1}
    \Norm{P_{1,j}\otimes\cdots\otimes P_{m,j}}I -
          P_{1,j}\otimes\cdots\otimes P_{m,j}
\]
is also $(\bS_1;\dots;\bS_m)$-separable.
Naturally, each of the identity terms
$\Norm{P_{1,j}\otimes\cdots\otimes P_{m,j}}I$
in this sum may be replaced by
$\norm{Q}I$ without compromising
$(\bS_1;\dots;\bS_m)$-separability,
from which it follows that
\( \br{n+1}\norm{Q}I-Q \) is also $(\bS_1;\dots;\bS_m)$-separable.

To prove item \ref{item:identity-sep:3}, apply Theorem \ref{thm:sep-bound} to obtain
$X=X^+-X^-$ where $X^\pm$ are $(\bS_1;\dots;\bS_m)$-separable with
\( \norm{X^\pm}\leq 2^{m-1}\sqrt{n}\fnorm{X}. \)
By item \ref{item:identity-sep:2}, it holds that
\( \br{n+1}\norm{X^+} I - X^+ \)
is $\pa{\bS_1;\dots;\bS_m}$-separable, implying that
\[ 2^{m-1}\sqrt{n}\cBr{n+1}\Fnorm{X} I - X^+ \]
is also $\pa{\bS_1;\dots;\bS_m}$-separable.
To complete the proof, it suffices to note that adding $X^-$ to this operator
yields another $\pa{\bS_1;\dots;\bS_m}$-separable operator.
\end{proof}

\section{Ball Around the Completely Noisy Channel}
\label{sec:balls}

All the technical results required to establish a ball of LOSR operations around
the completely noisy channel were proven in Section \ref{sec:gen}.
This section introduces proper formalization of the relevant notions of quantum
information so that the existence of the ball of LOSR operations may be stated
and proven rigorously.

\subsection{Linear Algebra: More Review and More Notation}

Finite-dimensional complex Euclidean spaces $\mathbb{C}^n$ are denoted by
capital script letters such as $\cX$, $\cY$, and $\cZ$.
The (complex) space of linear operators $A:\cX\to\cX$ is denoted $\lin{\cX}$ and
$\her{\cX}\subset\lin{\cX}$ denotes the (real) subspace of Hermitian
operators within $\lin{\cX}$.
In keeping with the notational convention of Section \ref{sec:gen},
$\pos{\cX}\subset\her{\cX}$ denotes the cone of positive semidefinite
operators within $\her{\cX}$.
The identity operator in $\lin{\cX}$ is denoted $I_\cX$, and the subscript is
dropped whenever the space $\cX$ is clear from the context.

A \emph{super-operator} is a linear operator of the form
$\Phi : \lin{\cX}\to\lin{\cA}$.
The identity super-operator from $\lin{\cX}$ to itself is denoted
$\idsup{\cX}$---again, the subscript is dropped at will.
A super-operator is said to be
\begin{itemize}

\item
\emph{positive on $\bK$} if
$\Phi\pa{X}$ is positive semidefinite whenever $X\in\bK$.

\item
\emph{positive} if $\Phi$ is positive on $\pos{\cX}$.

\item
\emph{completely positive} if $\Pa{\Phi \otimes \idsup{\cZ}}$ is positive
for every choice of complex Euclidean space $\cZ$.

\item
\emph{trace-preserving} if $\tr{\Phi\pa{X}}=\tr{X}$ for all $X$.

\end{itemize}
The \emph{Choi-Jamio\l kowski representation} of a super-operator
$\Phi:\lin{\cX}\to\lin{\cA}$ is defined by
\[
  \jam{\Phi} \defeq \sum_{i,j=1}^{\dim\pa{\cX}}
  \Phi\pa{e_ie_j^*} \otimes e_ie_j^*
\]
where $\set{e_1,\dots,e_{\dim\pa{\cX}}}$ denotes the standard basis of $\cX$.
The operator $\jam{\Phi}$ is an element of $\lin{\cA\otimes\cX}$.
Indeed, the mapping $\Jamiolkowski$ is an isomorphism between super-operators
and $\lin{\cA\otimes\cX}$.
The Choi-Jamio\l kowski representation has several convenient properties.
For example, $\Phi$ is completely positive if and only if $\jam{\Phi}$ is
positive semidefinite and $\Phi$ is trace-preserving if and only if it meets the
partial trace condition
\[ \ptr{\cA}{\jam{\Phi}}=I_\cX. \]
An additional property of the Choi-Jamio\l kowski representation is given in
Proposition \ref{prop:identities} of Section \ref{sec:char:LOSE}.

The Kronecker product shorthand notation
for operators $\kprod{X}{i}{j}$ and spaces $\kprod{\bS}{i}{j}$
of Section \ref{sec:gen} is extended in the obvious way to
complex Euclidean spaces $\kprod{\cX}{i}{j}$ and
super-operators $\kprod{\Phi}{i}{j}$.

\subsection{Formal Definitions of LOSE and LOSR Operations}
\label{sec:balls:defs}

Associated with each $d$-level physical system is a $d$-dimensional complex
Euclidean space $\cX=\mathbb{C}^d$.
The \emph{quantum state} of such a system at some fixed point in time is
described uniquely by a positive semidefinite operator $\rho\in\pos{\cX}$ with
$\ptr{}{\rho}=1$.
A \emph{quantum operation} is a physically realizable discrete-time mapping
(at least in an ideal sense)
that takes as input a quantum state of some system
$\cX$ and produces as output a quantum state of some system $\cA$.
Every quantum operation is represented uniquely by a completely positive and
trace-preserving super-operator $\Phi:\lin{\cX}\to\lin{\cA}$.

A quantum operation
$\Lambda : \lin{\kprod{\cX}{1}{m}} \to \lin{\kprod{\cA}{1}{m}}$
is a \emph{$m$-party LOSE operation with finite entanglement}
if there exist complex Euclidean spaces
$\cE_1,\dots,\cE_m$, a quantum state
$\sigma\in\pos{\kprod{\cE}{1}{m}}$, and quantum operations
$\Psi_i : \lin{\cX_i\otimes\cE_i}\to\lin{\cA_i}$ for each $i=1,\dots,m$ such that
\[ \Lambda\pa{X} = \Pa{\kprod{\Psi}{1}{m}} \pa{X\otimes\sigma} \]
for all $X\in\lin{\kprod{\cX}{1}{m}}$.
The operation $\Lambda$ is a
\emph{finitely approximable $m$-party LOSE operation}
if it lies in the closure of the set of
$m$-party LOSE operations with finite entanglement.

As mentioned in the introduction,
the need to distinguish between LOSE operations with finite entanglement and
finitely approximable LOSE operations arises from the fact
that there exist LOSE operations---such
as the example appearing in Ref.~\cite{LeungT+08}---that
cannot be implemented with any finite amount of shared entanglement, yet can be
approximated to arbitrary precision by LOSE operations with finite entanglement.

An analytic consequence of this fact is that the set of LOSE
operations with finite entanglement is not a closed set.
Fortunately, the work of the present paper is unhindered by a more encompassing
notion of LOSE operations that includes the closure of that set.
As such, the term ``LOSE operation'' is used to refer to any finitely
approximable LOSE operation; the restriction to finite entanglement is made
explicit whenever it is required.

An \emph{$m$-party LOSR operation}
is just a LOSE operation with finite entanglement in which the shared state
$\sigma\in\pos{\kprod{\cE}{1}{m}}$ is
$m$-partite separable.
(That is, the operator $\sigma$ is
$\Pa{\her{\cE_1};\dots;\her{\cE_m}}$-separable.)
An equivalent definition for LOSR operations is established by the following elementary proposition.

\begin{prop}
\label{prop:LOSR-convex}

A quantum operation
$\Lambda : \lin{\kprod{\cX}{1}{m}}\to\lin{\kprod{\cA}{1}{m}}$
is an $m$-party LOSR operation if and only if $\Lambda$ can be written as a convex combination of product super-operators of the form $\kprod{\Phi}{1}{m}$ where each
$\Phi_i:\lin{\cX_i}\to\lin{\cA_i}$ is a quantum operation for $i=1,\dots,m$.

\end{prop}

\begin{proof}

Let $\sigma$ be a $\Pa{\her{\cE_1};\dots;\her{\cE_m}}$-separable state and let
$\Psi_i:\lin{\cX_i\ot\cE_i}\to\lin{\cA_i}$ be quantum operations such that
$\Lambda\pa{X}=\Pa{\kprod{\Psi}{1}{m}}\pa{X\ot\sigma}$ for all $X$.
Let
\[ \sigma=\sum_{j=1}^n p_j \cBr{\sigma_{1,j}\ot\cdots\ot\sigma_{m,j}} \]
be a decomposition of $\sigma$ into a convex combination of product states, where each $\sigma_{i,j}\in\pos{\cE_i}$.
For each $i$ and $j$ define a quantum operation
\( \Phi_{i,j} : \lin{\cX_i}\to\lin{\cA_i} : X\mapsto\Psi_i\pa{X\ot\sigma_{i,j}} \)
and observe that
\[ \Lambda = \sum_{j=1}^n p_j \cBr{\Phi_{1,j}\ot\cdots\ot\Phi_{m,j}}. \]

Conversely, suppose that $\Lambda$ may be decomposed into a convex combination of product quantum operations as above.
Let $\cE_1,\dots,\cE_m$ be complex Euclidean spaces of dimension $n$ and let
$\set{e_{i,1},\dots,e_{i,n}}$ be an orthonormal basis of $\cE_i$ for each $i$.
Let
\[ \sigma=\sum_{j=1}^n p_j \cBr{ e_{1,j}e_{1,j}^* \ot\cdots\ot e_{m,j}e_{m,j}^*} \]
be a $\Pa{\her{\cE_1};\dots;\her{\cE_m}}$-separable state, let
\[
  \Psi_i:\lin{\cX_i\ot\cE_i}\to\lin{\cA_i} :
  X\ot E\mapsto\sum_{j=1}^n e_{i,j}^*Ee_{i,j} \cdot \Phi_{i,j}\pa{X}
\]
be quantum operations, and observe that
\( \Lambda\pa{X} = \Pa{\kprod{\Psi}{1}{m}}\pa{X\ot\sigma} \) for all $X$.
\end{proof}

By analogy with finitely approximable LOSE operations, one may
consider finitely approximable LOSR operations.
However, an easy application of Carath\'eodory's Theorem implies that any
LOSR operation can be implemented with finite randomness in such a way that each
of the spaces $\cE_1,\dots,\cE_m$ has dimension bounded (loosely) above by
$\dim\Pa{\her{\kprod{\cA}{1}{m}\otimes\kprod{\cX}{1}{m}}}$.
As suggested by this contrast, the sets of LOSE and LOSR operations
are not equal.
Indeed, while every LOSR operation is clearly a LOSE operation, much has been
made of the fact that there exist LOSE operations that are not LOSR
operations---this is quantum nonlocality.

\subsection{Ball of LOSR Operations Around the Completely Noisy Channel}

In the preliminary discussion of Section \ref{sec:gen} it was informally argued
that there exist subspaces
$\bQ_1,\dots,\bQ_m$
of Hermitian operators with the property that a quantum operation
$\Lambda$ is a LOSR operation if and only if
$\jam{\Lambda}$ is $\Pa{\bQ_1;\dots;\bQ_m}$-separable.
It was then established via Theorem \ref{thm:identity-sep} that any operator in
the product space $\kprod{\bQ}{1}{m}$ and close enough to the identity is
necessarily a $\Pa{\bQ_1;\dots;\bQ_m}$-separable operator, implying the
existence of a ball of LOSR operations around the completely noisy channel.
This last implication is addressed in Remark \ref{rem:ball:noisy} below---for
now, it remains only to identify the subspaces $\bQ_1,\dots,\bQ_m$.

Toward that end, recall from
Proposition \ref{prop:LOSR-convex}
that a quantum operation
$\Lambda : \lin{\kprod{\cX}{1}{m}} \to \lin{\kprod{\cA}{1}{m}}$
is a LOSR operation if and only if
it can be written as a convex combination of product super-operators of the form
$\kprod{\Phi}{1}{m}$.
Here each
$\Phi_i:\lin{\cX_i}\to\lin{\cA_i}$ is a quantum operation, so that
$\jam{\Phi_i}$ is positive semidefinite and satisfies
$\ptr{\cA_i}{\jam{\Phi_i}}=I_{\cX_i}$.
The set of all operators $X$ obeying the inhomogeneous linear condition
$\ptr{\cA_i}{X}=I_{\cX_i}$ is not a vector space.
But this set is easily extended to a unique smallest vector space by including
its closure under multiplication by real scalars, and it shall soon become
apparent that doing so poses no additional difficulty.

With this idea in mind, let
\[
  \bQ_i \defeq
  \Set{
    X\in\her{\cA_i\otimes\cX_i} : \ptr{\cA_i}{X}=\lambda I_{\cX_i}
    \textrm{ for some } \lambda\in\mathbb{R}
  }
\]
denote the subspace of operators
$X=\jam{\Phi}$ for which $\Phi:\lin{\cX_i}\to\lin{\cA_i}$
is a trace-preserving super-operator, or a scalar multiple thereof.
Clearly, $\jam{\Lambda}$ is $\Pa{\bQ_1;\dots;\bQ_m}$-separable whenever $\Lambda$ is a
LOSR operation.
Conversely, it is a simple exercise to verify that any quantum operation
$\Lambda$ for which $\jam{\Lambda}$ is $\Pa{\bQ_1;\dots;\bQ_m}$-separable is necessarily a
LOSR operation.

Letting
$n=\dim\pa{\kprod{\bQ}{1}{m}}$ and
$k = 2^{m-1}\sqrt{n}\cBr{n+1}$,
Theorem \ref{thm:identity-sep} states that the operator
$k\fnorm{A}I-A$ is $\Pa{\bQ_1;\dots;\bQ_m}$-separable for all
$A\in\kprod{\bQ}{1}{m}$.
In particular, $I-A$ is $\Pa{\bQ_1;\dots;\bQ_m}$-separable whenever
$\fnorm{A}\leq\frac{1}{k}$.
The following theorem is now proved.

\begin{theorem}[Ball around the completely noisy channel]
\label{thm:ball}

  Suppose $A\in\kprod{\bQ}{1}{m}$ satisfies $\fnorm{A} \leq \frac{1}{k}$.
  Then $I-A=\jam{\Lambda}$ where
  $\Lambda : \lin{\kprod{\cX}{1}{m}} \to \lin{\kprod{\cA}{1}{m}}$
  is an unnormalized LOSR operation.

\end{theorem}

\begin{remark}
\label{rem:ball:noisy}

As mentioned in the introduction, the (unnormalized) completely noisy channel
$\noisy:X\mapsto\tr{X}I$ satisfies $\jam{\noisy}=I$.
It follows from Theorem \ref{thm:ball} that there is a ball of LOSR operations
centred at the (normalized) completely noisy channel.
In particular, letting
$d=\dim\pa{\kprod{\cA}{1}{m}}$,
this ball has radius $\frac{1}{kd}$ in Frobenius norm.

\end{remark}

\begin{remark}
\label{rem:ball:same}

The ball of Theorem \ref{thm:ball} consists entirely of LOSE operations that are
also LOSR operations.
However, there seems to be no obvious way to obtain a bigger ball if
such a ball is allowed to contain operations that are LOSE but not LOSR.
Perhaps a more careful future investigation will uncover such a ball.

\end{remark}

\begin{remark}
\label{rem:balls:noise}

The ball of Theorem \ref{thm:ball} is contained within the product space
$\kprod{\bQ}{1}{m}$,
which is a strict subspace of the space
spanned by all quantum operations
$\Phi:\lin{\kprod{\cX}{1}{m}} \to \lin{\kprod{\cA}{1}{m}}$.
Why was attention restricted to this subspace?
The answer is that there are no LOSE or LOSR operations $\Lambda$ for which
$\jam{\Lambda}$ lies outside $\kprod{\bQ}{1}{m}$.
In other words, $\kprod{\bQ}{1}{m}$ is the \emph{largest} possible space in which
to find a ball of LOSR operations.
This fact follows from the discussion in Section \ref{sec:no-sig},
wherein it is shown that $\kprod{\bQ}{1}{m}$
is generated by the so-called \emph{no-signaling} quantum operations.

Of course, there exist quantum operations arbitrarily close to the
completely noisy channel that are not no-signaling operations,
much less LOSE or LOSR operations.
This fact might seem to confuse the study of, say, the effects of noise on such
operations because a completely general model of noise would allow for
extremely tiny perturbations that nonetheless turn no-signaling operations into
signaling operations.
This confusion might even be exacerbated by the fact that separable quantum
states, by contrast, are resilient to \emph{arbitrary} noise: any conceivable
physical perturbation of the completely mixed state is separable, so long as the
perturbation has small enough magnitude.

There is, of course, nothing unsettling about this picture.
In any reasonable model of noise, perturbations to a LOSE or LOSR operation
occur only on the local operations performed by the parties involved,
or perhaps on the state they share.
It is easy to see that realistic perturbations such as these always maintain the
no-signaling property of these operations.
Moreover, any noise \emph{not} of this form could, for example, bestow
faster-than-light communication upon spatially separated parties.

\end{remark}

\section{Recognizing LOSE Operations is NP-hard}
\label{sec:hard}

In this section the existence of the ball of LOSR operations around the
completely noisy channel is employed to prove that the weak membership problem
for LOSE operations is strongly NP-hard.
Informally, the weak membership problem asks,
\begin{quote}
  ``Given a description of a quantum operation $\Lambda$
  and an accuracy parameter $\varepsilon$,
  is $\Lambda$ within distance $\varepsilon$ of a LOSE operation?''
\end{quote}

This result is achieved in several stages.
Section \ref{sec:app:games} reviews a relevant recent result of
Kempe \emph{et al.}
pertaining to quantum games.
In Section \ref{sec:app:validity} this result is exploited in order to prove
that the weak validity
problem---a relative of the weak membership problem---is strongly NP-hard for
LOSE operations.
Finally, Section \ref{sec:app:membership} illustrates how the strong
NP-hardness of the weak membership problem for LOSE operations follows
from a Gurvits-Gharibian-style application
of Liu's version
of the Yudin-Nemirovski\u\i{} Theorem.
It is also noted that similar NP-hardness results hold trivially for LOSR
operations, due to the fact that separable quantum states arise as a special
case of LOSR operations in which the input space is empty.

\subsection{Co-Operative Quantum Games with Shared Entanglement}
\label{sec:app:games}

Local operations with shared entanglement have been studied in the context of
two-player co-operative games.
In these games, a referee prepares a question for each player and the players each respond to the referee with an answer.
The referee evaluates these answers and declares that the players have jointly won or lost the game according to this evaluation.
The goal of the players, then, is to coordinate their answers so as to maximize the probability with which the referee declares them to be winners.
In a \emph{quantum} game the questions and answers are quantum states.

In order to differentiate this model from a one-player game, the players are not permitted to communicate with each other after the referee has sent his questions.
The players can, however, meet prior to the commencement of the game in order to agree on a strategy.
In a quantum game the players might also prepare a shared entangled quantum state so as to enhance the coordination of their answers to the referee.

More formally, a \emph{quantum game} $G=(q,\pi,\mathbf{R},\mathbf{V})$ is specified by:
\begin{itemize}
  \item
  A positive integer $q$ denoting the number of distinct questions.
  \item
  A probability distribution $\pi$ on the question indices $\{1,\dots,q\}$, according to which 
  the referee selects his questions.
  \item
  Complex Euclidean spaces $\cV,\cX_1,\cX_2,\cA_1,\cA_2$ corresponding to the different quantum systems used by the referee and players.
  \item
  A set $\mathbf{R}$ of quantum states
  $\mathbf{R} = \{ \rho_i \}_{i=1}^q \subset \pos{\cV\otimes\cX_1\otimes\cX_2}$.
  These states correspond to questions and are selected by the referee according to $\pi$.
  \item
  A set $\mathbf{V}$ of unitary operators $\mathbf{V} = \{V_i \}_{i=1}^q \subset
  \lin{\cV\otimes\cA_1\otimes\cA_2}$.
  These unitaries are used by the referee to evaluate the players' answers.
\end{itemize}
For convenience, the two players are called \emph{Alice} and \emph{Bob}.
The game is played as follows.
The referee samples $i$ according to $\pi$ and prepares the state
$\rho_i\in\mathbf{R}$,
which is placed in the three quantum registers corresponding to
$\cV\otimes\cX_1\otimes\cX_2$.
This state contains the questions to be sent to the players:
the portion of $\rho_i$ corresponding to $\cX_1$ is sent to Alice,
the portion of $\rho_i$ corresponding to $\cX_2$ is sent to Bob, and
the portion of $\rho_i$ corresponding to $\cV$ is kept by the referee as a
private workspace.
In reply, Alice sends a quantum register corresponding to $\cA_1$ to the referee,
as does Bob to $\cA_2$.
The referee then applies the unitary operation $V_i\in \mathbf{V}$
to the three quantum registers corresponding to $\cV\otimes\cA_1\otimes\cA_2$,
followed by a standard measurement
$\{\Pi_\mathrm{accept},\Pi_\mathrm{reject}\}$
that dictates the result of the game.

As mentioned at the beginning of this subsection,
Alice and Bob may not communicate once the game commences.
But they may meet prior to the commencement of the game in order prepare a
shared entangled quantum state $\sigma$.
Upon receiving the question register corresponding to $\cX_1$ from the referee,
Alice may perform any physically realizable quantum operation upon that register
and upon her portion of $\sigma$.
The result of this operation shall be contained in the quantum register
corresponding to $\cA_1$---this is the answer that Alice sends to the referee.
Bob follows a similar procedure to obtain his own answer register corresponding
to $\cA_2$.

For any game $G$, the \emph{value} $\omega(G)$ of $G$ is the supremum of the
probability with which the referee can be made to accept taken over all
strategies of Alice and Bob.

\begin{theorem}[Kempe \emph{et al.}~\cite{KempeK+07}]
  There is a fixed polynomial $p$ such that the following promise problem is
  NP-hard under mapping (Karp) reductions:
  \begin{description}
  \item[Input.]
    A quantum game $G=(q,\pi,\mathbf{R},\mathbf{V})$.
    The distribution $\pi$ and the sets $\mathbf{R},\mathbf{V}$ are each given
    explicitly:
    for each $i=1,\dots,q$, the probability $\pi(i)$ is given in binary,
    as are the real and complex parts of each entry of the matrices
    $\rho_i$ and $V_i$.
  \item[Yes.] The value $\omega(G)$ of the game $G$ is 1.
  \item[No.]  The value $\omega(G)$ of the game $G$ is less than
    $1-\frac{1}{p(q)}$.
  \end{description}
\end{theorem}

\subsection{Strategies and Weak Validity}
\label{sec:app:validity}

Viewing the two players as a single entity, a quantum game may be seen as a
two-message quantum interaction between the referee and the players---a message
from the referee to the players, followed by a reply from the players to the
referee.
The actions of the referee during such an interaction are completely specified
by the parameters of the game.

In the language of Ref.~\cite{GutoskiW07}, the game specifies a
\emph{two-turn measuring strategy} for the referee.%
\footnote{
  Actually, the game specifies a \emph{one}-turn measuring \emph{co}-strategy
  for the referee.
  But this co-strategy can be re-written as a two-turn measuring strategy by a
  careful choice of input and output spaces.
  These terminological details are not important to purpose of this paper.
}
This strategy has a Choi-Jamio\l kowski representation given by some positive
semidefinite operators
\[
  R_\mathrm{accept},R_\mathrm{reject} \in
  \pos{\kprod{\cA}{1}{2}\ot\kprod{\cX}{1}{2}},
\]
which are easily computed given the parameters of the game.

In these games, the players implement a
one-turn non-measuring strategy compatible with
$\{R_\mathrm{accept},R_\mathrm{reject}\}$.
The Choi-Jamio\l kowski representation of this strategy is given by a
positive semidefinite operator
\[
  P \in
  \pos{\kprod{\cA}{1}{2}\ot\kprod{\cX}{1}{2}}.
\]
For any fixed strategy $P$ for the players, the probability with which they
cause the referee to accept is given by the
inner product
\[
  \Pr[\textrm{Players win with strategy $P$}] = \inner{R_\mathrm{accept}}{P}.
\]

In any game, the players combine to implement some physical operation
\( \Lambda:\lin{\kprod{\cX}{1}{2}}\to\lin{\kprod{\cA}{1}{2}}. \)
It is clear that a given super-operator $\Lambda$ denotes a legal strategy for the
players if and only if $\Lambda$ is a LOSE operation.
For the special case of one-turn non-measuring strategies---such as that of the
players---the Choi-Jamio\l kowski representation of such a strategy is
given by the simple formula $P=\jam{\Lambda}.$

Thus, the problem studied by Kempe \emph{et al.}~\cite{KempeK+07} of deciding
whether $\omega\pa{G}=1$ can be reduced via the formalism of strategies to an
optimization problem over the set of LOSE operations:
\[
  \omega\pa{G} = \sup_{\textrm{$\Lambda\in$ LOSE}}
  \left\{
    \inner{R_\mathrm{accept}}{\jam{\Lambda}}
  \right\}.
\]
The following theorem is thus proved.

\begin{theorem}
\label{thm:wval}

  The weak validity problem for the set of LOSE [LOSR] operations is
  strongly NP-hard [NP-complete] under mapping (Karp) reductions:
  \begin{description}
  \item[Input.]
    A Hermitian matrix $R$, a real number $\gamma$, and a positive
    real number $\varepsilon>0$.
    The number $\gamma$ is given explicitly in binary, as are the
    real and complex parts of each entry of $R$.
    The number $\varepsilon$ is given in unary, where $1^s$ denotes $\varepsilon=1/s$.
  \item[Yes.]
    There exists a LOSE [LOSR] operation $\Lambda$ such that
    $\inner{R}{\jam{\Lambda}} \geq \gamma + \varepsilon$.
  \item[No.]
    For every LOSE [LOSR] operation $\Lambda$ we have
    $\inner{R}{\jam{\Lambda}} \leq \gamma - \varepsilon$.
  \end{description}

\end{theorem}

\begin{remark}
\label{rem:sep-reduction}

The hardness result for LOSR operations follows from a simple reduction
from separable quantum states to LOSR operations:
every separable state may be written as a LOSR operation in which the input space has dimension one.
That weak validity for LOSR operations is in NP
(and is therefore NP-complete)
follows from the fact that all LOSR operations may be implemented with
polynomially-bounded shared randomness.

\end{remark}

\subsection{The Yudin-Nemirovski\u\i{} Theorem and Weak Membership}
\label{sec:app:membership}

Having established that the weak \emph{validity} problem for LOSE operations
is strongly NP-hard, the next step is to follow the leads of
Gurvits and Gharibian~\cite{Gurvits02,Gharibian08}
and apply Liu's version~\cite{Liu07} of the
Yudin-Nemirovski\u\i{} Theorem~\cite{YudinN76, GrotschelL+88}
in order to prove that the weak \emph{membership} problem for LOSE operations
is also strongly NP-hard.

The Yudin-Nemirovski\u\i{} Theorem establishes an oracle-polynomial-time
reduction from the weak validity problem to the weak membership problem for any
convex set $C$ that satisfies certain basic conditions.
One consequence of this theorem is that if the weak validity problem for $C$ is
NP-hard then the associated weak membership problem for $C$ is also
NP-hard.
Although hardness under mapping reductions is preferred, any hardness result
derived from the Yudin-Nemirovski\u\i{} Theorem in this way is only guaranteed
to hold under more inclusive oracle reductions.

The basic conditions that must be met by the set $C$ in order for the
Yudin-Nemirovski\u\i{} Theorem to apply are
(i) $C$ is bounded,
(ii) $C$ contains a ball, and
(iii) the size of the bound and the ball are polynomially related to the
dimension of the vector space containing $C$.
It is simple to check these criteria against the set of LOSE operations.
For condition (i), an explicit bound is implied by the following proposition.
The remaining conditions then follow from Theorem \ref{thm:ball}.

\begin{prop}
\label{prop:explicit-bound}

  For any quantum operation $\Phi:\lin{\cX}\to\lin{\cA}$ it holds that
  $\tnorm{\jam{\Phi}}=\dim(\cX)$.

\end{prop}

Here $\tnorm{X}\defeq\ptr{}{\sqrt{X^*X}}$ denotes the standard
\emph{trace norm} for operators.
Proposition \ref{prop:explicit-bound} implies a bound in terms of the
Frobenius norm via the inequality $\fnorm{X}\leq\tnorm{X}$.

\begin{proof}[Proof of Proposition \ref{prop:explicit-bound}]

  It follows from the definition of the Choi-Jamio\l kowski representation that
  \[
    \jam{\Phi} = \Pa{\Phi\ot\idsup{\cX}}(vv^*)
    \quad \textrm{ for } \quad
    v=\sum_{i=1}^{\dim(\cX)} e_i\ot e_i
  \]
  where $\set{e_1,\dots,e_{\dim(\cX)}}$ denotes the standard basis of $\cX$.
  As $\jam{\Phi}$ is positive semidefinite, it holds that
  \( \tnorm{\jam{\Phi}}=\ptr{}{\jam{\Phi}}. \)
  As $\Phi$  is trace-preserving, it holds that
  \( \ptr{}{\jam{\Phi}} = v^*v = \dim(\cX). \)
\end{proof}
%

The following theorem is now proved.
(As in Remark \ref{rem:sep-reduction},
the analogous result for LOSR operations follows from a straightforward
reduction from separable quantum states.)

\begin{theorem} \label{thm:wmem}

  The weak membership problem for the set of LOSE [LOSR] operations is
  strongly NP-hard [NP-complete] under oracle (Cook) reductions:
  \begin{description}

  \item[Input.]
    A Hermitian matrix $X\in\kprod{\bQ}{1}{m}$
    and a positive real number $\varepsilon>0$.
    The real and complex parts of each entry of $X$ are given explicitly in binary.
    The number $\varepsilon$ is given in unary, where $1^s$ denotes $\varepsilon=1/s$.

  \item[Yes.]
    $X=\jam{\Lambda}$ for some LOSE [LOSR] operation $\Lambda$.

  \item[No.]
    \( \norm{X-\jam{\Lambda}}\geq \varepsilon \)
    for every LOSE [LOSR] operation $\Lambda$.
    
  \end{description}

\end{theorem}

\section{Characterizations of Local Quantum Operations}
\label{sec:char}

Characterizations of LOSE and LOSR operations are presented in this section.
These characterizations are reminiscent of the well-known characterizations of
bipartite and multipartite separable quantum states due to
Horodecki \emph{et al.}~\cite{HorodeckiH+96,HorodeckiH+01}.
Specifically, it is proven in Section \ref{sec:char:LOSR} that $\Lambda$ is a LOSR
operation if and only if $\varphi\pa{\jam{\Lambda}}\geq 0$ whenever the linear
functional $\varphi$ is positive on the cone of $\Pa{\bQ_1;\dots;\bQ_m}$-separable
operators.
This characterization of LOSR operations is a straightforward application of the
fundamental Separation Theorems of convex analysis.
(See, for example, Rockafellar \cite{Rockafellar70} for proofs of these
theorems.)

More interesting is the characterization of LOSE operations presented in Section
\ref{sec:char:LOSE}:
$\Lambda$ is a LOSE operation if and only if
$\varphi\pa{\jam{\Lambda}}\geq 0$ whenever the linear functional $\varphi$ is
\emph{completely} positive on that \emph{same cone} of $\Pa{\bQ_1;\dots;\bQ_m}$-
separable
operators.
Prior to the present work, the notion of complete positivity was only ever
considered in the context wherein the underlying cone is the positive
semidefinite cone.
Indeed, what it even \emph{means} for a super-operator or functional to be
``completely'' positive on some cone other than the positive semidefinite cone
must be clarified before for any nontrivial discussion can occur.

The results of this section do not rely upon the results in previous sections,
though much of the notation of those sections is employed here.
To that notation we add the following.
For any operator $X$, its transpose $X^\trans$ and complex conjugate
$\overline{X}$ are always taken with respect to the standard basis.

\subsection{Characterization of Local Operations with Shared Randomness}
\label{sec:char:LOSR}

The characterization of LOSR operations presented herein is an immediate
corollary of the following simple proposition.

\begin{prop} \label{prop:LOSR}

Let $K\subset\mathbb{R}^n$ be any closed convex cone.
A vector $x\in\mathbb{R}^n$ is an element of $K$ if and only if
$\varphi\pa{x}\geq 0$ for every linear functional
$\varphi : \mathbb{R}^n \to \mathbb{R}$
that is positive on $K$.

\end{prop}

\begin{proof}

The ``only if'' part of the proposition is immediate:
as $x$ is in $K$, any linear functional positive on $K$ must also be positive on
$x$.
For the ``if'' part of the proposition, suppose that
$x$ is not an element of $K$.
The Separation Theorems from convex analysis imply that
there exists a vector $h\in\mathbb{R}^n$
such that
$\inner{h}{y} \geq 0$ for all $y\in K$, yet
$\inner{h}{x} < 0$.
Let $\varphi : \mathbb{R}^n \to \mathbb{R}$
be the linear functional given by
$\varphi : z \mapsto \inner{h}{z}$.
It is clear that $\varphi$ is positive on $K$, yet $\varphi\pa{x}<0$.
\end{proof}

\begin{corollary}[Characterization of LOSR operations] \label{thm:LOSR}

  A quantum operation
  $\Lambda : \lin{\kprod{\cX}{1}{m}} \to \lin{\kprod{\cA}{1}{m}}$
  is an $m$-party LOSR operation if and only if
  $\varphi\pa{\jam{\Lambda}} \geq 0$ for every linear functional
  $\varphi : \lin{\kprod{\cA}{1}{m}\otimes\kprod{\cX}{1}{m}} \to \mathbb{C}$
  that is positive on the cone of $\Pa{\bQ_1;\dots;\bQ_m}$-separable operators.

\end{corollary}

\begin{proof}

In order to apply Proposition \ref{prop:LOSR}, it suffices to note the following:
\begin{itemize}
\item
The space
$\her{\kprod{\cA}{1}{m}\ot\kprod{\cX}{1}{m}}$
is isomorphic to $\mathbb{R}^n$ for
$n=\dim\pa{\kprod{\cA}{1}{m}\ot\kprod{\cX}{1}{m}}^2$.
\item
The $\Pa{\bQ_1;\dots;\bQ_m}$-separable operators form a closed convex
cone within $\her{\kprod{\cA}{1}{m}\ot\kprod{\cX}{1}{m}}$.
\end{itemize}

While Proposition \ref{prop:LOSR} only gives the desired result for real
linear functionals
$\varphi: \her{\kprod{\cA}{1}{m}\ot\kprod{\cX}{1}{m}} \to \mathbb{R}$,
it is trivial to construct a complex extension functional
$\varphi':\lin{\kprod{\cA}{1}{m}\otimes\kprod{\cX}{1}{m}} \to \mathbb{C}$
that agrees with $\varphi$ on $\her{\kprod{\cA}{1}{m}\ot\kprod{\cX}{1}{m}}$.
\end{proof}

\subsection{Characterization of Local Operations with Shared Entanglement}
\label{sec:char:LOSE}


\begin{notation}

  For each $i=1,\dots,m$ and each complex Euclidean space $\cE_i$,
  let $\bQ_i\pa{\cE_i}\subset\her{\cA_i\otimes\cX_i\otimes\cE_i}$
  denote the subspace of operators $\jam{\Psi}$ for which
  $\Psi:\lin{\cX_i\otimes\cE_i}\to\lin{\cA_i}$
  is a trace-preserving super-operator, or a scalar multiple thereof.

\end{notation}

\begin{theorem}[Characterization of LOSE operations]
\label{thm:LOSE}

  A quantum operation
  $\Lambda : \lin{\kprod{\cX}{1}{m}} \to \lin{\kprod{\cA}{1}{m}}$
  is an $m$-party LOSE operation if and only if
  $\varphi\pa{\jam{\Lambda}} \geq 0$ for every linear functional
  $\varphi : \lin{\kprod{\cA}{1}{m}\otimes\kprod{\cX}{1}{m}} \to \mathbb{C}$
  with the property that the super-operator
  $\Pa{\varphi \otimes \idsup{\kprod{\cE}{1}{m}}}$
  is positive on the cone of
  $\Pa{\bQ_1\pa{\cE_1};\dots;\bQ_m\pa{\cE_m}}$-separable operators
  for all choices of complex Euclidean spaces $\cE_1,\dots,\cE_m$.

\end{theorem}

\begin{remark}
\label{rem:LOSE:cp}

The positivity condition of Theorem \ref{thm:LOSE} bears striking resemblance to
the familiar notion of complete positivity of a super-operator.
With this resemblance in mind, a linear functional $\varphi$ obeying
the positivity condition of Theorem \ref{thm:LOSE} is said to be
\emph{completely positive on the
$\Pa{\bQ_1\pa{\cE_1};\dots;\bQ_m\pa{\cE_m}}$-separable family of cones}.
In this sense, Theorem \ref{thm:LOSE} represents what seems,
to the knowledge of the author,
to be the first application of the notion of complete positivity to a family of
cones other than the positive semidefinite family.

Moreover, for any fixed choice of complex Euclidean spaces
$\cE_1,\dots,\cE_m$ there exists a linear functional $\varphi$ for which
$\Pa{\varphi \otimes \idsup{\kprod{\cE}{1}{m}}}$
is positive on the cone of
$\Pa{\bQ_1\pa{\cE_1};\dots;\bQ_m\pa{\cE_m}}$-separable operators,
and yet $\varphi$ is \emph{not} completely positive on this family of
cones.
This curious property is a consequence of the fact the set of LOSE
operations with finite entanglement is not a closed set.
By contrast, complete positivity (in the traditional sense) of a super-operator
$\Phi:\lin{\cX}\to\lin{\cA}$ is assured whenever $\Pa{\Phi\otimes\idsup{\cZ}}$ is
positive for a space $\cZ$ with dimension at least that of $\cX$.
(See, for example, Bhatia \cite{Bhatia07} for a proof of this fact.)

\end{remark}

The proof of Theorem \ref{thm:LOSE} employs the following helpful identity
involving the Choi-Jamio\l kowski representation for super-operators.
This identity may be verified by straightforward calculation.

\begin{prop}
\label{prop:identities}

  Let $\Psi:\lin{\cX\otimes\cE} \to \lin{\cA}$
  and let $Z\in\lin{\cE}$.
  Then the super-operator
  \(\Lambda : \lin{\cX} \to \lin{\cA} \)
  defined by
  \( \Lambda\pa{X} = \Psi \pa{X\otimes Z} \)
  for all $X$ satisfies
  \(
  \jam{\Lambda} =
  \Ptr{\cE}{ \cBr{ I_{\cA\otimes\cX} \otimes Z^\trans } \jam{\Psi} }.
  \)

\end{prop}

%

The following technical lemma is also employed in the proof of
Theorem \ref{thm:LOSE}.

\begin{lemma} \label{lm:LOSE:ip}

  Let \( \Psi : \lin{\cX\otimes\cE} \to  \lin{\cA} \),
  let \( Z \in \lin{\cE}$,
  and
  let \( \varphi : \lin{ \cA\otimes\cX } \to \mathbb{C} \).
  Then the super-operator  \( \Lambda : \lin{\cX} \to \lin{\cA} \) defined by
  \( \Lambda\pa{X} = \Psi \Pa{ X\otimes Z } \) for all $X$ satisfies
  \[
    \varphi\pa{\jam{\Lambda}} = \Inner{ \overline{Z} }
    { \Pa{ \varphi \otimes \idsup{\cE} } \Pa{ \jam{ \Psi } } }.
  \]

\end{lemma}

\begin{proof}

Let $H$ be the unique operator satisfying
$\varphi\pa{X}=\inner{H}{X}$ for all $X$
and note that the adjoint
$\varphi^*:\mathbb{C}\to\lin{ \cA\otimes\cX }$ satisfies
$\varphi^*(1)=H$.
Then
\begin{align*}
  \Inner{ \overline{Z} }{
    \Pa{ \varphi \otimes \idsup{\cE} }
    \Pa{ \jam{ \Psi } }
  }
  &= \Inner{ \varphi^*\pa{1} \otimes \overline{Z} }{ \jam{ \Psi } }
  = \Inner{ H \otimes \overline{Z} }{ \jam{ \Psi } } \\
  &= \Inner{ H }{
      \Ptr{\cE}
        { \cBr{ I_{ \cA\otimes\cX } \otimes Z^\trans } \jam{\Psi} }
    }
  = \varphi\pa{\jam{\Lambda}}.
\end{align*}
\end{proof}


\begin{proof}[Proof of Theorem \ref{thm:LOSE}]

For the ``only if'' part of the theorem,
let $\Lambda$ be any LOSE operation with finite entanglement and let
$\Psi_1,\dots,\Psi_m,\sigma$ be such that
$\Lambda : X \mapsto \Pa{\kprod{\Psi}{1}{m}}\pa{X\otimes \sigma}$.
Let $\varphi$ be any linear functional on
$\lin{\kprod{\cA}{1}{m}\otimes\kprod{\cX}{1}{m}}$
that satisfies the stated positivity condition.
Lemma \ref{lm:LOSE:ip} implies
\[
  \varphi\pa{\jam{\Lambda}} = \Inner{ \overline{\sigma} }{
    \Pa{\varphi \otimes \idsup{\kprod{\cE}{1}{m}} }
    \Pa{ \jam{\kprod{\Psi}{1}{m}} }
  } \geq 0.
\]
A standard continuity argument establishes the desired implication when
$\Lambda$ is a finitely approximable LOSE operation.

For the ``if'' part of the theorem, suppose that
$\Xi : \lin{\kprod{\cX}{1}{m}} \to \lin{\kprod{\cA}{1}{m}}$ is a
quantum operation that is \emph{not} a LOSE operation.
The Separation Theorems from convex analysis imply that there
is a Hermitian operator
$H\in\her{\kprod{\cA}{1}{m}\otimes\kprod{\cX}{1}{m}}$ such that
$\inner{H}{\jam{\Lambda}} \geq 0$ for all LOSE operations $\Lambda$, yet
$\inner{H}{\jam{\Xi}} < 0$.
Let $\varphi : \lin{\kprod{\cA}{1}{m}\otimes\kprod{\cX}{1}{m}}\to\mathbb{C}$
be the linear functional given by
$\varphi : X \mapsto \inner{H}{X}$.

It remains to verify that $\varphi$ satisfies the desired positivity condition.
Toward that end, let $\cE_1,\dots,\cE_m$ be arbitrary complex Euclidean spaces.
By convexity, it suffices to consider only those
$\Pa{\bQ_1\pa{\cE_1};\dots;\bQ_m\pa{\cE_m}}$-separable operators that are product
operators.
Choose any such operator and note that, up to a scalar multiple,
it has the form
$\jam{ \kprod{\Psi}{1}{m} }$ where each
$\Psi_i : \lin{\cX_i\otimes\cE_i} \to \lin{\cA_i}$ is a quantum operation.
The operator
\[
  \Pa{ \varphi \otimes \idsup{\kprod{\cE}{1}{m}} }
  \Pa{ \jam{ \kprod{\Psi}{1}{m} } }
\]
is positive semidefinite if and only if it has a nonnegative inner product with
every density operator in $\pos{\kprod{\cE}{1}{m}}$.
As $\sigma$ ranges over all such operators, so does $\overline{\sigma}$.
Moreover, every such $\sigma$---together with $\Psi_1,\dots,\Psi_m$---induces a
LOSE operation $\Lambda$ defined by
$\Lambda : X \mapsto \Pa{\kprod{\Psi}{1}{m}}\pa{X\otimes \sigma}$.
Lemma \ref{lm:LOSE:ip} and the choice of $\varphi$ imply
\[
0 \leq \varphi \pa{\jam{\Lambda}} =
\Inner{ \overline{\sigma} }{
  \Pa{\varphi \otimes \idsup{\kprod{\cE}{1}{m}} }
  \Pa{ \jam{\kprod{\Psi}{1}{m}} }
},
\]
and so $\varphi$ satisfies the desired positivity condition.
\end{proof}

\section{No-Signaling Operations} \label{sec:no-sig}

It was claimed in Remark \ref{rem:balls:noise} of Section \ref{sec:balls}
that the product space
$\kprod{\bQ}{1}{m}$
is spanned by Choi-Jamio\l kowski representations of no-signaling operations.
It appears as though this fact has yet to be noted explicitly in the literature,
so a proof is offered in this section.
More accurately, two characterizations of no-signaling operations are presented
in Section \ref{sec:no-sig:char}, each of which is expressed as a condition on
Choi-Jamio\l kowski representations of super-operators.
The claim of Remark \ref{rem:balls:noise} then follows immediately from these
characterizations.
Finally, Section \ref{sec:no-sig:counter-example} provides
an example of a so-called \emph{separable} no-signaling operation that is
not a LOSE operation, thus ruling out an easy ``short cut'' to the ball of LOSR
operations revealed in Theorem \ref{thm:ball}.

\subsection{Formal Definition of a No-Signaling Operation}

Intuitively, a quantum operation $\Lambda$ is no-signaling if it cannot be used by
spatially separated parties to violate relativistic causality.
Put another way, an operation $\Lambda$ jointly implemented by several parties is
no-signaling if those parties cannot use $\Lambda$ as a ``black box'' to
communicate with one another.

In order to facilitate a formal definition for no-signaling operations,
the shorthand notation for Kronecker products from Section \ref{sec:gen} must be
extended:
if $K\subseteq\set{1,\dots,m}$ is an arbitrary index set with
$K=\set{k_1,\dots,k_n}$ then we write
\[ \cX_K \defeq \cX_{k_1} \otimes \cdots \otimes \cX_{k_n} \]
with the convention that $\cX_\emptyset = \mathbb{C}$.
As before, a similar notation also applies to operators, sets of operators, and
super-operators.
The notation $\overline{K}$ refers to the set of indices \emph{not} in $K$, so that $K$,$\overline{K}$ is a partition of $\set{1,\dots,m}$.

Formally then, a quantum operation
$\Lambda:\lin{\kprod{\cX}{1}{m}}\to\lin{\kprod{\cA}{1}{m}}$ is an
\emph{$m$-party no-signaling operation} if for each index set
$K\subseteq\set{1,\dots,m}$ we have
\( \ptr{\cA_K}{\Lambda\pa{\rho}} = \ptr{\cA_K}{\Lambda\pa{\sigma}} \)
whenever
\( \ptr{\cX_K}{\rho} = \ptr{\cX_K}{\sigma}. \)

What follows is a brief argument that this condition captures the meaning of a
no-signaling operation---a more detailed discussion of this condition can be
found in Beckman \emph{et al.}~\cite{BeckmanG+01}.
If $\Lambda$ is no-signaling and $\rho,\sigma$ are locally indistinguishable to a
coalition $K$ of parties
(for example, when \( \ptr{\cX_{\overline{K}}}{\rho} = \ptr{\cX_{\overline{K}}}{\sigma} \))
then clearly the members of $K$ cannot perform a measurement on their portion of
the output that might allow them to distinguish $\Lambda\pa{\rho}$ from
$\Lambda\pa{\sigma}$
(that is, \( \ptr{\cA_{\overline{K}}}{\Lambda\pa{\rho}} = \ptr{\cA_{\overline{K}}}{\Lambda\pa{\sigma}} \)).
For otherwise, the coalition $K$ would have extracted
information---a signal---that would allow it to distinguish
$\rho$ from $\sigma$.

Conversely, if there exist input states $\rho,\sigma$ such that
\( \ptr{\cX_{\overline{K}}}{\rho} = \ptr{\cX_{\overline{K}}}{\sigma} \) and yet
\( \ptr{\cA_{\overline{K}}}{\Lambda\pa{\rho}} \neq \ptr{\cA_{\overline{K}}}{\Lambda\pa{\sigma}} \)
then there exists a measurement that allows the coalition $K$ to distinguish
these two output states with nonzero bias, which implies that signaling must
have occurred.

It is not hard to see that every LOSE operation is also a no-signaling
operation.
Conversely, much has been made of the fact that there exist no-signaling
operations that are not LOSE operations---this is so-called ``super-strong''
nonlocality, exemplified by the popular ``nonlocal box'' discussed in Section
\ref{sec:no-sig:counter-example}.

\subsection{Two Characterizations of No-Signaling Operations}
\label{sec:no-sig:char}

In this section it is shown that the product space $\kprod{\bQ}{1}{m}$ is spanned
by Choi-Jamio\l kowski representations of no-signaling operations.
(Recall that each $\bQ_i\subset\her{\cA_i\otimes\cX_i}$ denotes the subspace of
Hermitian operators $\jam{\Phi}$ for which $\Phi:\lin{\cX_i}\to\lin{\cA_i}$ is a
trace-preserving super-operator, or a scalar multiple thereof.)
Indeed, that fact is a corollary of the following characterizations of
no-signaling operations.

\begin{theorem}[Two characterizations of no-signaling operations]
\label{thm:char:no-signal}

Let
$\Lambda:\lin{\kprod{\cX}{1}{m}}\to\lin{\kprod{\cA}{1}{m}}$
be a quantum operation.
The following are equivalent:
\begin{enumerate}

\item \label{item:no-signal}
$\Lambda$ is a no-signaling operation.

\item \label{item:prod-space}
$\jam{\Lambda}$ is an element of $\kprod{\bQ}{1}{m}$.

\item \label{item:constraints}
For each index set $K\subseteq\set{1,\dots,m}$
there exists
an operator
$Q\in\pos{\cA_{\overline{K}}\ot\cX_{\overline{K}}}$
with
\( \ptr{\cA_K}{\jam{\Lambda}} = Q\otimes I_{\cX_K}. \)

\end{enumerate}

\end{theorem}

\begin{remark}

Membership in the set of no-signaling operations may be verified in
polynomial time by checking the linear constraints in
Item \ref{item:constraints} of Theorem \ref{thm:char:no-signal}.
While the number of such constraints grows exponentially with $m$, 
this exponential growth is not a problem because the number of parties
$m$ is always $O\pa{\log n}$ for $n=\dim\pa{\kprod{\bQ}{1}{m}}$.
(This logarithmic bound follows from the fact that each space $\bQ_i$ has dimension at
least two and the total dimension $n$ is the product of the dimensions of each
of the $m$ different spaces.)

\end{remark}

The partial trace condition of Item \ref{item:constraints} of
Theorem \ref{thm:char:no-signal} is quite plainly
suggested by Ref.~\cite[Theorem 6]{GutoskiW07}.
Moreover, essential components of the proofs presented for two of the three
implications claimed in Theorem \ref{thm:char:no-signal}
appear in a 2001 paper of
Beckman \emph{et al.}~\cite{BeckmanG+01}.
The following theorem, however, establishes the third implication and appears to
be new.

\begin{theorem} \label{thm:TP-constraints}

A Hermitian operator $X\in\her{\kprod{\cA}{1}{m}\otimes\kprod{\cX}{1}{m}}$ is in
$\kprod{\bQ}{1}{m}$ if and only if for each index set
$K\subseteq\set{1,\dots,m}$ there exists a Hermitian operator
$Q\in\her{\cA_{\overline{K}}\ot\cX_{\overline{K}}}$
with
\( \ptr{\cA_K}{X} = Q \otimes I_{\cX_K}. \)

\end{theorem}

\begin{proof}
The ``only if'' portion of the theorem is straightforward---only the
``if'' portion is proven here.
The proof proceeds by induction on $m$.
The base case $m=1$ is trivial.
Proceeding directly to the general case,
let $s=\dim\pa{\her{\cA_{m+1}\otimes\cX_{m+1}}}$ and
let $\set{E_1,\dots,E_s}$ be a basis of
$\her{\cA_{m+1}\otimes\cX_{m+1}}$.
Let $X_1,\dots,X_s\in\her{\kprod{\cA}{1}{m}\otimes\kprod{\cX}{1}{m}}$
be the unique operators satisfying
\[ X = \sum_{j=1}^s X_j \otimes E_j. \]
It shall be proven that $X_1,\dots,X_s\in\kprod{\bQ}{1}{m}$.
The intuitive idea is to exploit the linear independence of $E_1,\dots,E_s$ in
order to ``peel off'' individual product terms in the decomposition of $X$.

Toward that end, for each fixed index $j\in\set{1,\dots,s}$ let
$H_j$ be a Hermitian operator
for which the real number $\inner{H_j}{E_i}$ is nonzero only when
$i=j$.
Define a linear functional
$\varphi_j:E\mapsto\inner{H_j}{E}$
and note that
\[
\Pa{ \idsup{\kprod{\cA}{1}{m}\otimes\kprod{\cX}{1}{m} } \ot \varphi_j}
  \pa{X} =
\sum_{i=1}^s \varphi_j\pa{E_i} X_i = \varphi_j\pa{E_j} X_j.
\]
Fix an arbitrary partition $K,\overline{K}$ of the index set $\set{1,\dots,m}$.
By assumption,
\( \ptr{\cA_K}{X} = Q\otimes I_{\cX_K} \)
for some Hermitian operator $Q$.
Apply
$\trace_{\cA_K}$
to both sides of the above identity, then use the fact that
$\trace_{\cA_K}$
and $\varphi_j$ act upon different spaces
to obtain
\begin{align*}
& \varphi_j\pa{E_j} \ptr{\cA_K}{X_j} \\
={}& \Ptr{\cA_K}{
  \Pa{\idsup{ \kprod{\cA}{1}{m}\otimes\kprod{\cX}{1}{m} } \ot \varphi_j}
  \Pa{X}
} \\
={}& \Pa{ 
  \idsup{\cA_{\overline{K}} \otimes \kprod{\cX}{1}{n}}
  \ot \varphi_j
} \Pa{ \Ptr{\cA_K}{X} } \\
={}& \Pa{
  \idsup{\cA_{\overline{K}}\otimes \cX_{\overline{K}} }
  \ot \varphi_j
} \pa{ Q }
\otimes I_{\cX_K},
\end{align*}
from which it follows that $\ptr{\cA_K}{X_j}$ is a product operator of the form
$R \otimes I_{\cX_K}$ for some Hermitian operator $R$.
As this identity holds for all index sets $K$,
it follows from the induction hypothesis that
$X_j\in\kprod{\bQ}{1}{m}$ as desired.

Now, choose a maximal linearly independent subset
$\set{X_1,\dots,X_t}$ of $\set{X_1,\dots,X_s}$ and let
$Y_1,\dots,Y_t$ be the unique Hermitian
operators satisfying
\[ X = \sum_{i=1}^t X_i \otimes Y_i. \]
A similar argument shows $Y_1,\dots,Y_t\in\bQ_{m+1}$,
which completes the induction.
\end{proof}


\begin{proof}[Proof of Theorem \ref{thm:char:no-signal}]

\begin{description}

\item[Item \ref{item:constraints} implies item \ref{item:prod-space}.]

This implication follows immediately from Theorem \ref{thm:TP-constraints}.

\item[Item \ref{item:prod-space} implies item \ref{item:no-signal}.]

The proof of this implication borrows heavily from the proof of Theorem 2 in
Beckman \emph{et al.}~\cite{BeckmanG+01}.

Fix any partition $K,\overline{K}$ of the index set $\set{1,\dots,m}$.
Let $s=\dim\pa{\lin{\cX_{\overline{K}}}}$ and $t=\dim\pa{\lin{\cX_K}}$
and let $\set{\rho_1,\dots,\rho_s}$ and $\set{\sigma_1,\dots,\sigma_t}$
be bases of $\lin{\cX_{\overline{K}}}$ and $\lin{\cX_K}$, respectively,
that consist entirely of density operators.
Given any two operators $X,Y\in\lin{\kprod{\cX}{1}{m}}$ let
$x_{j,k},y_{j,k}\in\mathbb{C}$ be the unique coefficients of $X$ and $Y$
respectively in the product basis $\set{\rho_j\otimes \sigma_k}$.
Then
\( \ptr{\cX_K}{X} = \ptr{\cX_K}{Y} \)
implies
\[ \sum_{k=1}^t x_{j,k} = \sum_{k=1}^t y_{j,k} \]
for each fixed index $j=1,\dots,s$.

As $\jam{\Lambda}\in\kprod{\bQ}{1}{m}$, it is possible to write
\[
\jam{\Lambda} = \sum_{l=1}^n \jam{\Phi_{1,l}}\ot\cdots\ot\jam{\Phi_{m,l}}
\]
where $n$ is a positive integer and
$\Phi_{i,l}:\lin{\cX_i}\to\lin{\cA_i}$ satisfies
$\jam{\Phi_{i,l}}\in\bQ_i$ for each of the indices
$i=1,\dots,m$ and $l=1,\dots,n$.
In particular, as each $\Phi_{i,l}$ is (a scalar multiple of) a trace-preserving
super-operator, it holds that for each index $l=1,\dots,n$ there exists
$a_l\in\mathbb{R}$ with
$\tr{\Phi_{K,l}\pa{\sigma}}=a_l$ for all density operators $\sigma$.
Then
\begin{align*}
  \ptr{\cA_K}{\Lambda\pa{X}}
  &= \sum_{l=1}^n a_l \cdot \cBr{\sum_{k=1}^t x_{j,k}} \cdot
    \sum_{j=1}^s  \Phi_{\overline{K},l}\pa{\rho_j} \\
  &= \sum_{l=1}^n a_l \cdot \cBr{\sum_{k=1}^t y_{j,k}} \cdot
    \sum_{j=1}^s  \Phi_{\overline{K},l}\pa{\rho_j}
  = \ptr{\cA_K}{\Lambda\pa{Y}}
\end{align*}
as desired.

\item[Item \ref{item:no-signal} implies item \ref{item:constraints}.]

This implication is essentially a multi-party generalization of Theorem 8 in
Beckman \emph{et al.}~\cite{BeckmanG+01}.
The proof presented here differs from theirs in some interesting but
non-critical details.

Fix any partition $K,\overline{K}$ of the index set $\set{1,\dots,m}$.
To begin,  observe that
\[
\ptr{\cX_K}{X} = \ptr{\cX_K}{Y}
\implies
\ptr{\cA_K}{\Lambda\pa{X}} = \ptr{\cA_K}{\Lambda\pa{Y}}
\]
for \emph{all} operators
$X,Y\in\lin{\kprod{\cX}{1}{m}}$---not just density operators.
(This observation follows from the fact that $\lin{\kprod{\cX}{1}{m}}$ is
spanned by the density operators---a fact used in the above proof that
item \ref{item:prod-space} implies item \ref{item:no-signal}.)

Now, let $s=\dim\pa{\cX_{\overline{K}}}$ and $t=\dim\pa{\cX_K}$
and let
$\set{e_1,\dots,e_s}$ and $\set{f_1,\dots,f_t}$
be the standard bases of $\cX_{\overline{K}}$ and $\cX_K$ respectively.
If $c$ and $d$ are distinct indices in $\set{1,\dots,t}$
and $Z\in\lin{\cX_{\overline{K}}}$ is any operator then
\[
  \ptr{ \cX_K }{ Z \otimes f_cf_d^* } =
  Z \otimes \tr{ f_cf_d^* } =
  0_{\cX_{\overline{K}}} = \ptr{ \cX_K }{ 0_{\kprod{\cX}{1}{m}} }
\]
and hence
\[
  \Ptr{ \cA_K }{ \Lambda\Pa{ Z \otimes f_cf_d^*   } } =
  \Ptr{ \cA_K }{ \Lambda\Pa{ 0_{\kprod{\cX}{1}{m}} } } =
  \Ptr{ \cA_K }{ 0_{\kprod{\cA}{1}{m}} } = 0_{\cA_{\overline{K}}}.
\]
(Here a natural notation for the zero operator was used implicitly.)
Similarly, if $\rho\in\lin{\cX_K}$ is any density operator then
\[
\Ptr{ \cA_K }{ \Lambda\Pa{ Z \otimes f_cf_c^* } } =
\Ptr{ \cA_K }{ \Lambda\Pa{ Z \otimes \rho } }
\]
for each fixed index $c=1,\dots,t$.
Employing these two identities, one obtains
\begin{align*}
  \Ptr{ \cA_K }{ \jam{\Lambda} }
  &= \sum_{a,b=1}^s \sum_{c=1}^t \Ptr{ \cA_K }{
    \Lambda\Pa{e_ae_b^* \otimes f_cf_c^*}} \otimes \cBr{e_ae_b^* \otimes f_cf_c^*} \\
  &= \sum_{a,b=1}^s  \Ptr{ \cA_K }{
    \Lambda\Pa{e_ae_b^* \otimes \rho}} \otimes
    e_ae_b^* \otimes \cBr{\sum_{c=1}^tf_cf_c^*}
   = \jam{\Psi} \otimes I_{ \cX_K }
\end{align*}
where the quantum operation $\Psi$ is defined by
$\Psi:X\mapsto\ptr{ \cA_K }{\Lambda\Pa{X\otimes\rho}}$.
As $\jam{\Psi} \otimes I_{ \cX_K }$ is a product operator of the desired form,
the proof that item \ref{item:no-signal} implies item
\ref{item:constraints} is complete.

\end{description}
\end{proof}

\subsection{A Separable No-Signaling Operation that is Not a LOSE Operation}
\label{sec:no-sig:counter-example}

One of the goals of the present work is to establish a ball of LOSR operations
around the completely noisy channel.
As such, it is prudent to consider the possibility of a short cut.
Toward that end, suppose $\Lambda:\lin{\cX_1\otimes\cX_2}\to\lin{\cA_1\otimes\cA_2}$
is a quantum operation for which the operator $\jam{\Lambda}$
is $\Pa{\her{\cA_1\otimes\cX_1};\her{\cA_2\otimes\cX_2}}$-separable.
Quantum operations with separable Choi-Jamio\l kowski representations such as
this are called \emph{separable operations} \cite{Rains97}.
If an operation is both separable and no-signaling then must it always be a
LOSE operation, or even a LOSR operation?
An affirmative answer to this question would open a relatively simple path to
the goal by permitting the use of existing knowledge of separable balls around
the identity operator.

Alas, such a short cut is not to be had:
there exist no-signaling operations $\Lambda$ that are not LOSE operations, yet
$\jam{\Lambda}$ is separable.
One example of such an operation is the so-called ``nonlocal box'' discovered in
1994 by Popescu and Rohrlich~\cite{PopescuR94}.
This nonlocal box is easily formalized as a two-party no-signaling quantum
operation $\Lambda$, as in Ref.~\cite[Section V.B]{BeckmanG+01}.
That formalization is reproduced here.

Let \( \cX_1 = \cX_2 = \cA_1 = \cA_2 = \mathbb{C}^2 \),
let $\set{e_0,e_1}$ denote the standard bases of both $\cX_1$ and $\cA_1$, and
let $\set{f_0,f_1}$ denote the standard bases of both $\cX_2$ and $\cA_2$.
Write
\[ \rho_{ab} \defeq e_ae_a^*\otimes f_bf_b^* \]
for $a,b\in\set{0,1}$.
The nonlocal box $\Lambda:\lin{\kprod{\cX}{1}{2}}\to\lin{\kprod{\cA}{1}{2}}$
is defined by
\begin{align*}
\Set{ \rho_{00},\rho_{01},\rho_{10} }
&\stackrel{\Lambda}{\longmapsto}
\frac{1}{2} \cBr{ \rho_{00} + \rho_{11} } \\
\rho_{11}
&\stackrel{\Lambda}{\longmapsto}
\frac{1}{2} \cBr{ \rho_{01} + \rho_{10} }.
\end{align*}
Operators not in $\spn\set{\rho_{00},\rho_{01},\rho_{10},\rho_{11}}$ are
annihilated by $\Lambda$.
It is routine to verify that $\Lambda$ is a no-signaling operation, and
this operation $\Lambda$ is known not to be a LOSE operation~\cite{PopescuR94}.
To see that $\jam{\Lambda}$ is separable, write
\begin{align*}
  E_{a\to b} &\defeq e_be_a^* \\
  F_{a\to b} &\defeq f_bf_a^*
\end{align*}
for $a,b\in\set{0,1}$.
Then for all $X\in\lin{\kprod{\cX}{1}{2}}$ it holds that
\begin{align*}
\Lambda\pa{X}
&= \frac{1}{2}
\Big[ E_{0\to 0} \otimes F_{0\to 0} \Big] X
\Big[ E_{0\to 0} \otimes F_{0\to 0} \Big]^*
+  \frac{1}{2}
\Big[ E_{0\to 1} \otimes F_{0\to 1} \Big] X
\Big[ E_{0\to 1} \otimes F_{0\to 1} \Big]^* \\
&+ \frac{1}{2}
\Big[ E_{0\to 0} \otimes F_{1\to 0} \Big] X
\Big[ E_{0\to 0} \otimes F_{1\to 0} \Big]^*
+ \frac{1}{2}
\Big[ E_{0\to 1} \otimes F_{1\to 1} \Big] X
\Big[ E_{0\to 1} \otimes F_{1\to 1} \Big]^* \\
&+ \frac{1}{2}
\Big[ E_{1\to 0} \otimes F_{0\to 0} \Big] X
\Big[ E_{1\to 0} \otimes F_{0\to 0} \Big]^*
+ \frac{1}{2}
\Big[ E_{1\to 1} \otimes F_{0\to 1} \Big] X
\Big[ E_{1\to 1} \otimes F_{0\to 1} \Big]^* \\
&+ \frac{1}{2}
\Big[ E_{1\to 0} \otimes F_{1\to 1} \Big] X
\Big[ E_{1\to 0} \otimes F_{1\to 1} \Big]^*
+ \frac{1}{2}
\Big[ E_{1\to 1} \otimes F_{1\to 0} \Big] X
\Big[ E_{1\to 1} \otimes F_{1\to 0} \Big]^*,
\end{align*}
from which the
$\Pa{\her{\cA_1\otimes\cX_1};\her{\cA_2\otimes\cX_2}}$-separability of $\jam{\Lambda}$
follows.
It is interesting to note that the nonlocal box is
the smallest possible nontrivial example of such an operation---the number of
parties $m=2$ and the input spaces
$\cX_1,\cX_2$ and output spaces $\cA_1,\cA_2$ all have dimension two.

\section{Open Problems}
\label{sec:conclusion}

Several interesting open problems are suggested by the present work:
\begin{description}

\item[Bigger ball of LOSE or LOSR operations.]

The size of the ball of (unnormalized) LOSR operations established in Theorem
\ref{thm:ball} scales as $\Omega\Pa{2^{-m} n^{-3/2}}$.
While the exponential dependence on the number of parties $m$ seems unavoidable,
might it be possible to improve the exponent from $-m$ to $-m/2$, as is the case
with separable quantum states~\cite{GurvitsB03}?
Can the exponent on the dimension $n$ be improved?

As mentioned in Remark \ref{rem:ball:same}, it is not clear that there is a
ball of LOSE operations that strictly contains any ball of LOSR operations.
Does such a larger ball exist?

\item[Completely positive super-operators.]

As noted in Remark \ref{rem:LOSE:cp},
the characterization of LOSE operations is interesting because it involves
linear functionals that are not just positive, but
``completely'' positive on the family of
$\Pa{\bQ_1\pa{\cE_1};\dots;\bQ_m\pa{\cE_m}}$-separable cones.

Apparently, the study of completely positive super-operators has until now been
strictly limited to the context of positive semidefinite input cones.
Any question that may be asked of conventional completely positive
super-operators might also be asked of this new class of completely positive
super-operators.

\item[Entanglement required for approximating LOSE operations.]

It was mentioned in the introduction that there exist LOSE operations that
cannot be implemented with any finite amount of shared entanglement
\cite{LeungT+08}.
The natural question, then, is how much entanglement is necessary to achieve an
arbitrarily close approximation to such an operation?

The present author conjectures that for every two-party LOSE operation $\Lambda$ there exists an
$\varepsilon$-approximation $\Lambda'$ of $\Lambda$
in which the dimension of the shared entangled state scales as
$O\pa{2^{\varepsilon^{-a}}n^b}$ for some positive constants $a$ and $b$
and some appropriate notion of $\varepsilon$-approximation.
Here $n=\dim\Pa{\kprod{\bQ}{1}{m}}$ is the dimension of the space in which
$\jam{\Lambda}$ lies.

Evidence in favor of this conjecture can be found in Refs.~\cite{CleveH+04,KempeR+07,LeungT+08}.
Moreover, the example in Ref.~\cite{LeungT+08} strongly suggests that the exponential dependence on $1/\varepsilon$ is unavoidable.
The pressing open question is whether the polynomial dependence on $n$ holds for the general case.
At the moment, \emph{no upper bound at all} is known for this general class of two-party LOSE operations.

\end{description}

\section*{Acknowledgements}

The author is grateful to Marco Piani and Stephanie Wehner
for pointers to relevant literature, and to
John Watrous for invaluable guidance.
This research was supported by graduate scholarships from Canada's NSERC and the
University of Waterloo.


\newcommand{\etalchar}[1]{$^{#1}$}

\end{document}